\documentclass[10pt,journal,compsoc]{IEEEtran}

\makeatletter
\def\paragraph{\@startsection{paragraph}{4}{2\parindent}{-0ex plus -0.1ex minus 
-0.1ex}%
{0ex}{\normalfont\normalsize\bfseries}}%
\makeatother

\newcommand{\eg}{e.g.\xspace}
\newcommand{\ie}{i.e.\xspace}

\usepackage[numbers]{natbib}
\newcommand{\citeN}[1]{\citet{#1}}
\usepackage{graphicx}
\usepackage{paralist}
\graphicspath{{./plots/}{./diagrams/}}
\DeclareGraphicsExtensions{.pdf}

\usepackage[cmex10]{amsmath}
\usepackage{amsthm}

\makeatletter

\theoremstyle{definition}
\newtheorem{defn}{\protect Definition}

\theoremstyle{plain}
\newtheorem*{thm*}{\protect Theorem}

\makeatother

\usepackage{array}

\usepackage[caption=false]{subfig}
\usepackage{dblfloatfix}

\usepackage{hyperref}
\usepackage{xspace}
\usepackage{varioref}
\usepackage[utf8]{inputenc}

\usepackage{xcolor}
\definecolor{dark-red}{rgb}{0.4,0.15,0.15}
\definecolor{dark-blue}{rgb}{0.15,0.15,0.4}
\definecolor{medium-blue}{rgb}{0,0,0.5}
\hypersetup{
  colorlinks, linkcolor={dark-red},
  citecolor={dark-blue}, urlcolor={medium-blue}
}

\usepackage{balance}

\usepackage[boxed]{algorithm2e}

\SetStartEndCondition{ }{}{}%
\SetKwProg{Fn}{def}{\string:}{}
\SetKwFunction{Range}{range}
\SetKw{KwTo}{in}\SetKwFor{For}{for}{\string:}{}%
\SetKwIF{If}{ElseIf}{Else}{if}{:}{elif}{else:}{}%
\SetKwFor{While}{while}{:}{}%
\SetKwComment{tcp}{\# }{}
\SetKwComment{tcc}{"""}{"""}
\SetCommentSty{textit}
\AlgoDontDisplayBlockMarkers\SetAlgoNoEnd\SetAlgoNoLine%
\DontPrintSemicolon

\newcommand{\bl}{\mathcal{L}}
\newcommand{\bv}{\mathcal{O}}
\newcommand{\be}{\mathcal{E}}

\newcommand{\plotwidth}{.8\columnwidth}

\newcommand{\threeplotwidth}{.32\textwidth}

\newcommand{\thetitle}{PSBS: Practical Size-Based Scheduling}

\usepackage{mathtools}
\DeclarePairedDelimiter\floor{\lfloor}{\rfloor}

\usepackage{comment}

\begin{document}

\title{\thetitle}
\author{Matteo~Dell'Amico,~Damiano~Carra,~and~Pietro~Michiardi%
\IEEEcompsocitemizethanks{\IEEEcompsocthanksitem M. Dell'Amico and P. Michiardi are with EURECOM, France.
\IEEEcompsocthanksitem D. Carra is with University of Verona, Italy.
}
}

\markboth{}{Dell'Amico \MakeLowercase{\textit{et al.}}: \thetitle}

\IEEEtitleabstractindextext{%
\begin{abstract}
Size-based schedulers have very desirable performance properties:
optimal or near-optimal response time can be coupled with strong
fairness. Despite this, however, such systems are 
rarely implemented in practical settings, because they require knowing
\textit{a priori} the amount of work needed to complete jobs: this
assumption is difficult to satisfy in concrete systems. It is
definitely more likely to inform the system with an \emph{estimate} of
the job sizes, but existing studies point to somewhat
pessimistic results if size-based policies use imprecise
job size estimations.

We take the goal of designing scheduling policies that 
\emph{explicitly deal with inexact job sizes}. First, we
prove that, in the absence of errors, it is always possible to improve
any scheduling policy by designing a size-based one that
\emph{dominates} it: in the new policy, \emph{no jobs} will complete
later than in the original one. Unfortunately,
size-based schedulers can perform badly with
inexact job size information when
job sizes are heavily skewed; we show that this issue, and the
pessimistic results shown in the literature, are due to problematic
behavior when large jobs are underestimated. Once the problem is
identified, it is possible to amend size-based
schedulers to solve the issue.

We generalize FSP -- a fair and efficient size-based scheduling policy
-- to solve the problem highlighted above; 
in addition, our solution deals with different job weights (that can be 
assigned to a job independently from its size). 
We provide an efficient
implementation of the resulting protocol, which we call
\emph{Practical Size-Based Scheduler} (PSBS).

Through simulations evaluated on synthetic and real workloads, we show
that PSBS has near-optimal performance in a large variety of cases
with inaccurate size information, that it performs fairly and that it
handles job weights correctly. We believe that this work shows
that PSBS is indeed pratical, and we maintain that it could inspire
the design of schedulers in a wide array of real-world use cases.
\end{abstract}

}

\maketitle

\IEEEdisplaynontitleabstractindextext

\IEEEpeerreviewmaketitle





\IEEEraisesectionheading{\section{Introduction}\label{sec:introduction}}

\IEEEPARstart{I}{n} computer systems, several mechanisms can be modeled as queues where
jobs (\eg, batch computations or data transfers) compete
to access a shared resource (\eg, processor or network). In this context,
size-based scheduling protocols, which prioritize jobs that are closest
to completion, are well known to have very desirable properties: the
shortest remaining processing time policy (SRPT) provides optimal mean
response time~\cite{schrage1966queue}, while the fair sojourn protocol
(FSP)~\cite{Friedman2003} provides similar efficiency while
guaranteeing strong fairness properties.

Despite these characteristics, however, scheduling policies similar to
SRPT or FSP are very rarely deployed in production: the \textit{de
facto} standard are size-oblivious policies similar to processor sharing
(PS), which divides resources evenly among jobs in the queue.  A key
reason is that, in real systems, the job size is almost never
known \textit{a priori}. It is, instead, often possible to provide
\emph{estimations} of job size, which may vary in precision depending
on the use case; however, the impact of errors due to these
estimations in realistic scenarios is not yet well understood.

Perhaps surprisingly, very few works tackled the problem of size-based
scheduling with inaccurate job size information: as we discuss more in
depth in Section~\ref{sec:related}, the existing literature gives
somewhat pessimistic results, suggesting that size-based scheduling is
effective only when the error on size estimation is small; known
analytical results depend on restrictive assumptions on size
estimations, while simulation-based analyses only cover a limited
family of workloads. More importantly, no study we are aware of
tackled the design of size-based schedulers that are
\emph{explicitly designed with the goal of coping with errors} in job
size information. Our endeavor is to create a
\emph{practical} size-based scheduling protocol, that has an efficient
implementation and handles imprecise size information.
In addition, the scheduler should allow setting \emph{weights} to jobs, to 
control the relative proportion of the resources assigned to them.  

In Section~\ref{sec:analytic}, we provide a proof that it is possible
to improve \emph{any} size-oblivious policy by simulating that policy
and running jobs sequentially in the order in which they complete in
the simulated policy. The resulting policy \emph{dominates} the
latter: \emph{no job} will complete later due to the policy
change. This result generalizes the known fact that FSP dominates
PS~\cite{Friedman2003} and gives strong fairness guarantees, but it
does not hold when job size information is not exact.

In Section~\ref{sec:errors}, we give a qualitative analysis of the
impact of size estimation errors on scheduling behavior: we
show that, for heavy-tailed job size distributions, size-based policies can
behave problematically when large jobs are under-estimated: this
phenomenon, indeed, explains the pessimistic results
observed in previous works.

Fortunately, it is possible to solve
the aforementioned problem: in Section~\ref{sec:solution}, we propose a scheduling protocol that drastically improves the
behavior of disciplines such as FSP and SRPT when estimation errors exist.
Our approach, which we call PSBS (Practical Size-Based Scheduler), is a generalization of FSP featuring an efficient $O(\log n)$ implementation and support for job weights.

We developed a simulator, described in Section~\ref{sec:simulator}, to
study the behavior of size-based and size-oblivious scheduling policies
in a wide variety of scenarios. Our simulator allows both replaying real traces and generating synthetic ones varying system load, job size distribution
and inter-arrival time distribution; for both synthetic and real
workloads, scheduling protocols are evaluated on errors that range
between relatively small quantities and others that may vary even by
orders of magnitude. The simulator is released as open-source software,
to help reproducibility of our results and to facilitate further
experimentation.

From the experimental results of
Section~\ref{sec:experimental_results}, we highlight the
following, validated both on synthetic and real traces:
\begin{enumerate}
  \item When job size is not heavily skewed, SRPT and FSP outperform
    size-oblivious disciplines even when job size estimation is very
    imprecise, albeit past work would hint towards important performance
    degradation; on the other hand, when the job size distribution is
    heavy-tailed, performance degrades noticeably;
  \item The scheduling disciplines we propose (from which we derive PSBS) 
    do not suffer from the
    performance issues of FSP and
    SRPT; they provide good performance for a large part of the
    parameter space that we explore, being outperformed by a processor
    sharing strategy only when \emph{both} the job size distribution is heavily
    skewed \emph{and} size estimations are very inaccurate;
  \item PSBS handles job weights correctly and behaves fairly, guaranteeing 
    that most jobs complete
    in an amount of time that is proportional to their size.
\end{enumerate}

As we discuss in Section~\ref{sec:conclusion}, we conclude that our
work highlights and solves a key weakness of size-based scheduling
protocols when size estimation errors are present; the fact that
PSBS consistently performs close to optimally highlights that
size-based schedulers are more viable in real systems than what was
known from the state of the art; we believe that our work can help
inspiring both the design of new size-based schedulers for real
systems and analytic research that can provide better insight on
scheduling when errors are present.

\section{Related Work}
\label{sec:related}

We discuss two main areas of related work: first, results for
size-based scheduling on single-server queues when job sizes are known 
only approximately; second,
practical approaches devoted to the estimation of job sizes.

\subsection{Single-Server Queues}
\label{sec:related-ssq}

Performance evaluation of scheduling policies in single-server queues
has been the subject of many studies in the last 40 years.  Most of
these works, however, focus on extreme situations: the size of
a given job is either \emph{completely unknown} or \emph{known
perfectly}. In the first (\emph{size-oblivious}) case,
smart scheduling choices can still be taken by considering the overall
job size distribution: for example, in the common case where job sizes
are skewed -- \ie, a small percentage of jobs are responsible for most
work performed in the system -- it is smart to give
priority to younger jobs, because they are likely to complete
faster. Least-Attained-Service (LAS)~\cite{rai2003analysis}, also
known in the literature as \emph{Foreground-Background}
(FB)~\cite{kleinrock1975theory} and \emph{Shortest Elapsed Time}
(SET)~\cite{coffman1973operating}, employs this principle. Similar
principles guide the design of multi-level
queues~\cite{kleinrock1976queueing,guo2002scheduling}.

When job size is known \emph{a priori}, scheduling policies taking
into account this information are well known to perform better (\eg,
obtain shorter response times) than size-oblivious
ones. Unfortunately, job sizes can often be only known approximately,
rather than exactly.  Since in our paper we consider this case, we
review the literature that targets this problem.

Perhaps due to the difficulty of providing analytical results, not much work considers the effect of inexact
job size information on size-based scheduling. 
\citeN{lu2004size} have been the first to consider this problem,
showing that size-based scheduling is useful only when job size
evaluations are reasonably good (high correlation, greater than 0.75,
between the real job size and its estimate). Their evaluation focuses
on a single heavy-tailed job size distribution, and does
  not explain the causes of the observed results. Instead, we show
the effect of different job size distributions (heavy-tailed,
memoryless and light-tailed), and we show how to modify the size-based
scheduling policies to make them robust to job estimation errors.

\citeN{wierman2008scheduling} provide analytical
results for a class of size-based policies, but consider an
impractical assumption: results depend on a bound on the estimation error.
In the common case where most estimations are close to the real value
but there are outliers, bounds need to be set according to outliers,
leading to pessimistic predictions on performance. In our work,
instead, we do not impose any bound on the error.
Semi-clairvoyant scheduling~\cite{Bender2002,Becchetti2004} is the
problem where the scheduler, rather than knowing precisely a job's
size $s$, knows its size class $\floor{\log_2\left(s\right)}$. It can
be regarded as similar to the bounded error case.

Other works examined the effect of imprecise size information in
size-based schedulers for web servers~\cite{harchol2003size} and
MapReduce~\cite{chang2011scheduling}. In both cases, these are
simulation results that are ancillary to the proposal of a scheduler
implementation for a given system, and they are limited to a single
type of workload.

To the best of our knowledge, these are the only works targeting job
size estimation errors in size-based scheduling. We remark that, by
using an experimental approach and replaying traces, we can take into
account phenomena that are difficult to consider in analytic approaches,
such as periodic temporal patterns or correlations
between job size and submission time.

\subsection{Job Size Estimation}

In the context of distributed 
  systems,
FLEX~\cite{wolf2010flex} and HFSP~\cite{pastorelli2013hfsp} proved
that size-based scheduling can perform well in practical
scenarios. In both cases, job size estimation is
  performed with very simple approaches (\ie, by sampling the execution
  time of a part of the job): such rough estimates are sufficient to
  provide good performance, and our results provide an explanation to this.

In several practical contexts, rough job size
  estimations are easy to perform. For instance, web servers can use
file size as an estimator of job
size~\cite{schroeder2006web}, and the variability of the end-to-end
transmission bandwidth determines the estimation error.
More elaborate ways to estimate size
  are often available, since job size
  estimation is useful in many domains; examples
  are approaches that deal with predicting the size of MapReduce
  jobs~\cite{ARIA11,nsdi12-c,query_perf} and of database
  queries~\cite{lipton1995query}. Estimation error can be always
evaluated \textit{a posteriori}, and this evaluation can be used to
decide if size-based scheduling works better than size-oblivious policies.

\section{\label{sec:analytic}Dominance Results With Known Job Sizes}

Friedman and Henderson have proven that FSP -- a policy
that executes jobs serially in the order in which they complete in PS
-- dominates PS: when job sizes are known exactly, no jobs complete
later in PS than in FSP~\cite{Friedman2003}. This is a strong fairness
guarantee, but in most practical cases a policy such as FSP falls short because
of its lack of configurability: for example, it does not allow to prioritize
jobs.
We show here that Friedman and Henderson's
results can be generalized: no matter what the original scheduling
policy is, it is possible to simulate it and execute jobs in the order
of their completion: the resulting policy will still dominate it. Our
PSBS policy, described in Section~\ref{sec:solution}, is an instance
of this set of policies which allows setting job priorities.

We consider here the single-machine scheduling problem with release
times and preemption. In this section, we consider the \emph{offline}
scheduling problem, where release times and sizes of each job are
known in advance. As we shall see in the following, PSBS (like FSP)
guarantees these dominance results while also being appliable
\emph{online}, \ie, without any information about jobs released in the
future.

Our goal, that materializes in the $\mathrm{Pri}$ 
scheduler, is to minimize the sum of 
completion times (using Graham et al.'s notation~\cite{
Graham1979}, the $1|r_{i};pmtn|\sum C_{i}$ problem) with the 
additional dominance requirement: no job should complete later than 
in a scheduler which is taken as a reference for fairness. Without 
this limitation, the optimal solution is the Shortest Remaining 
Processing Time (SRPT) policy. We call \emph{
schedule }a function $\omega\left(i,t\right)$
that outputs the fraction of system resources allocated to job $i$
at time $t$. For example, for the processor-sharing (PS) scheduler,
when $n$ jobs are \emph{pending }(released and not yet completed),
$\omega\left(i,t\right)=\frac{1}{n}$ if job $i$ is pending and $0$
otherwise. Furthermore, we call $C_{i,\omega}$ the completion time
of job $i$ under schedule $\omega$.
\begin{defn}
Schedule $\omega$ \emph{dominates} schedule $\omega'$ if $C_{i,\omega}\leq C_{i,\omega'}$
for each job $i$.
\end{defn}
Our scheduler prioritizes jobs according to the order in which they
complete in $\omega$: its \emph{completion sequence}.
\begin{defn}
A \emph{completion sequence $S=\left[s_{1},\ldots,s_{n}\right]$ }is
an ordering of the jobs to be scheduled. A schedule $\omega$ \emph{has
completion sequence} $S$ if $C_{s_{i},\omega}\leq C_{s_{j},\omega}\forall i<j$.
\end{defn}

\begin{defn}
For a completion sequence $S$, the $\mathrm{Pri}_{S}$ schedule is
such that $\mathrm{Pri}_{S}\left(i,t\right)=1$ if $i$ is the first
pending job to appear in $S$; $\mathrm{Pri}_{S}\left(i,t\right)=0$
otherwise.
\end{defn}
We now show that scheduling jobs in the order in which they complete
under $\omega'$ dominates $\omega$.
\begin{thm*}
$\mathrm{Pri}_{S}$ dominates any schedule with completion sequence
$S$.\end{thm*}
\begin{proof}
We have to show that $C_{i,\mathrm{Pri}_{S}}\leq C_{i,\omega}$ for
each job $i$ and any schedule $\omega$ with completion sequence
$S$. Let $j$ be the position of $i$ in $S$ (i.e., $i=s_{j}$); we call $M$ the minimal makespan of the $S_{\leq j}=\left\{ s_{1},\ldots,s_{j}\right\} $
set of jobs,%
\footnote{The \emph{makespan} of a set of jobs is the maximum among their completion
times, therefore $M=\min_{\omega\in\Omega}\max_{i\in\left\{ 1,\ldots,j\right\} }C_{S_{i},\omega}$
where $\Omega$ is the set of all possible schedules.%
} and we show that $C_{i,\mathrm{Pri}_{S}}\leq M$ and $M\leq C_{i,\omega}$:
\begin{itemize}
\item $C_{i,\mathrm{Pri}_{S}}\leq M$: minimizing the makespan of $S_{\leq j}$
is equivalent to solving the $1|r_{i};pmtn|C_{\max}$ problem applied
to the jobs in $S_{\leq j}$: this is guaranteed if all resources
are assigned to jobs in $S_{\leq j}$ as long as any of them are pending
\cite{Liu1973}. $\mathrm{Pri}_{S}$ guarantees this, hence the makespan
of $S_{\leq j}$ using $\mathrm{Pri}_{S}$ is $M$. Since $i\in S_{\leq j}$,
$C_{i,\mathrm{Pri}_{S}}\leq M$.
\item $M\leq C_{i,\omega}$ follows trivially from $\omega$
having completion sequence $S$ and, therefore, $C_{i,\omega}$ being the
makespan for $S_{\leq j}$ using schedule $\omega$.
\qedhere
\end{itemize}
\end{proof}
This theorem generalizes Friedman and Henderson's
results: FSP follows from applying $\mathrm{Pri}_{S}$
to the completion sequence of PS.
The generalization is important: in practice, one can define 
a scheduler that provides a desired type of fairness, and optimize
the performance in terms of completion time by applying the $\mathrm{Pri}_{S}$
scheduler. If the system deals with different 
classes of jobs that have different weights, we can take discriminatory
processor sharing (DPS) as a reference: our theorem guarantees that 
$\mathrm{Pri}_{S}$ dominates DPS. We have exploited
exactly this results in our PSBS scheduler, which, 
in the absence of errors, dominates DPS. Only when errors are present -- and
this dominance result does not apply -- PSBS deviates from the behavior of
$\mathrm{Pri}_{S}$.


\section{Scheduling Based on Estimated Sizes}
\label{sec:errors}

We now describe the effects that estimation errors have on existing
size-based policies such as SRPT and FSP.
We notice that under-estimation triggers a behavior
which is problematic for heavy-tailed job size
distributions: this is the key insight that will lead to the design of
PSBS.

\subsection{SRPT and FSP}
\label{sec:srpt_fsp}

SRPT gives priority to the job with smallest remaining
processing time. It is \emph{preemptive}: a new job with size
smaller than the remaining processing time of the running one will
preempt (\ie, interrupt) the latter. When the scheduler
has access to exact job sizes, SRPT has optimal mean sojourn time
(MST)~\cite{schrage1966queue} -- \emph{sojourn time}, or \emph{response
  time}, is the time that passes between a job's submission and its
completion.

SRPT may cause \emph{starvation} (\ie, never providing access to
resources): for example, if small jobs are constantly submitted, large
jobs may never get served; while this phenomenon appears rare in
practical cases~\cite{Bansal:2001:ASS:384268.378792}, it is
nevertheless worrying. FSP (also known as
\emph{fair queuing}~\cite{fair_queuing} and
\emph{Vifi}~\cite{gorinsky2007fair}) doesn't suffer
from starvation by virtue of \emph{job aging}: 
FSP serves the job that would complete earlier in a
\emph{virtual} emulated system running a processor sharing (PS)
discipline: since all jobs eventually complete in the virtual system,
they will also eventually be scheduled in the real one.

With no estimation errors, FSP provides a value of MST which is close
to what is provided by SRPT while guaranteeing fairness due to the
dominance result discussed in Section~\ref{sec:analytic}.
When errors are present, this
property is not guaranteed; however, our results in
Section~\ref{sec:fairness} show that FSP preserves better
fairness than SRPT also in this case.

\subsection{Dealing With Errors: SRPTE and FSPE}
\label{sec:under_over}

We now consider SRPT and FSP when the
  scheduler uses \emph{estimated} job sizes rather than exact
  ones. For clarity, we will refer hereinafter to \emph{SRPTE} and
  \emph{FSPE} in this case.

\begin{figure}[!t]
  \centering
  \includegraphics[width=\columnwidth]{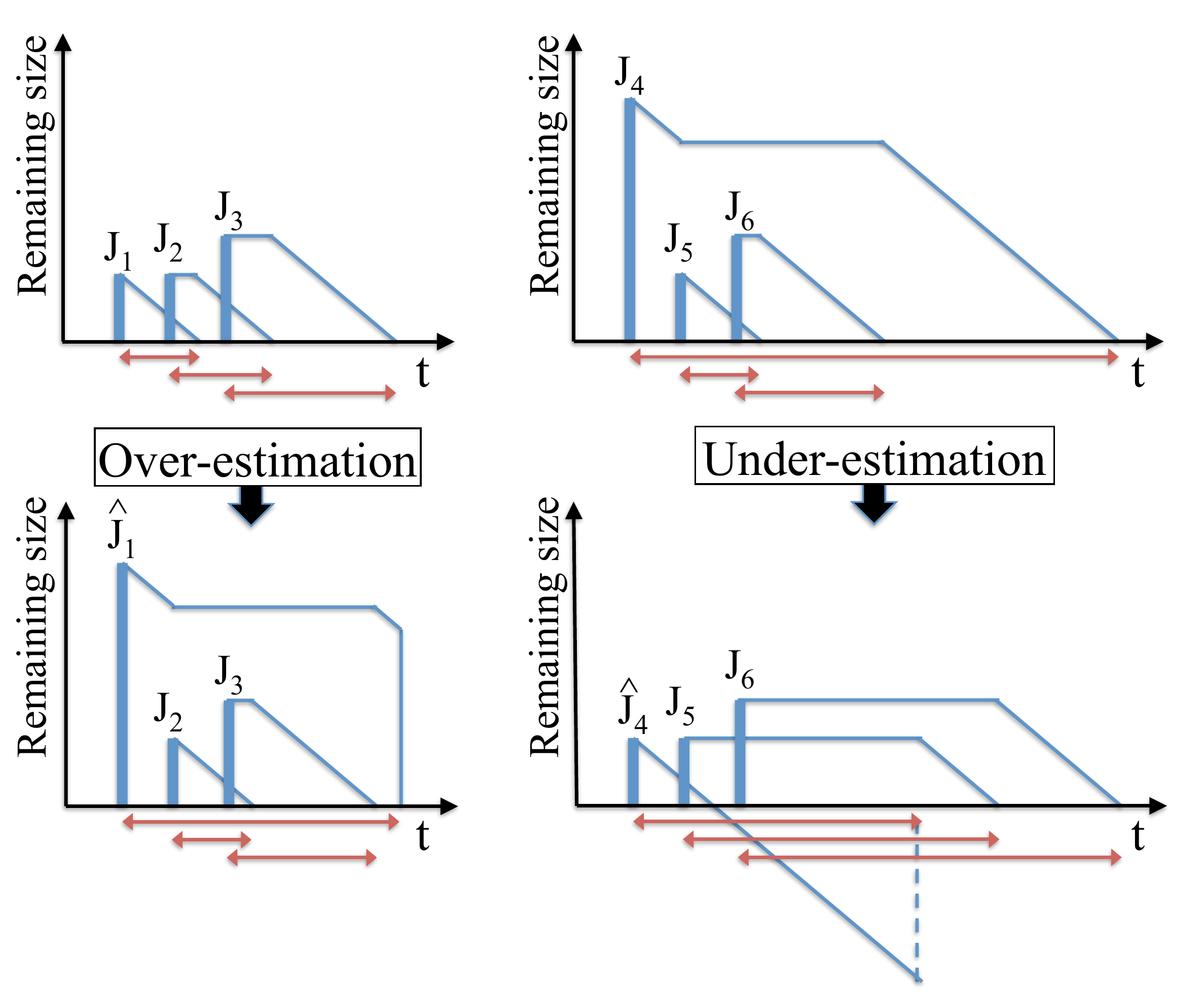}
  \caption{Examples of scheduling without (top) and with (bottom) errors.}
  \label{fig:under_over}
\end{figure}

In Fig.~\ref{fig:under_over}, we provide an illustrative example
where a single job size is over- or under-estimated while the others
are estimated correctly, focusing (because of its simplicity) on
SRPTE; sojourn times are represented by the horizontal arrows. The
left column of Fig.~\ref{fig:under_over} illustrates the effect of
over-estimation. In the top, we show how the scheduler behaves without
errors, while in the bottom we show what happens when the size of job
$J_1$ is over-estimated. The graphs shows the remaining (estimated)
processing time of the jobs over time, assuming a normalized service
rate of 1. Without errors, $J_2$ does not preempt $J_1$, and $J_3$
does not preempt $J_2$.  Instead, when the size of $J_1$ is
over-estimated, both $J_2$ and $J_3$ preempt $J_1$. Therefore, the only
penalized job (\ie, experiencing higher sojourn time) is the over-estimated one.
Jobs with smaller sizes are always able to
preempt an over-estimated job, therefore the basic property of SRPT
(favoring small jobs) is not significantly compromised.

The right column of Fig.~\ref{fig:under_over} illustrates the effect
of under-estimation. With no estimation
errors (top), a large job, $J_4$, is preempted by small ones ($J_5$ and
$J_6$). If the size of the large job is under-estimated (bottom), its
estimated remaining processing time eventually reaches zero: we call
\emph{late} a job with zero or negative estimated remaining processing
time. \emph{A late job cannot be preempted} by newly arrived jobs,
since their size estimation will always be larger than zero. In
practice, since preemption is inhibited, the under-estimated job
\emph{monopolizes the system} until its completion, impacting negatively all
 waiting jobs.

This phenomenon is particularly harmful with heavily skewed job sizes, if estimation errors are proportional to size:
if there are few very large jobs and many
  small ones, a single late large job can significantly delay several
  small ones, which will need to wait for the late job to complete for an amount of time which is disproportionate to their size
  before having an opportunity of being served.

Even if the impact of under-estimation seems straightforward to
understand, surprisingly \emph{no work in the literature has ever
discussed it}. To the best of our knowledge, we are the
first to identify this problem, which significantly influences
scheduling policies dealing with inaccurate job size.

In FSPE, the phenomena we observe are analogous: job
  size over-estimation delays only the over-estimated job;
  under-estimation can result in jobs terminating in the virtual PS
  queue before than in the real system; this is impossible in absence
  of errors due to the dominance result introduced in
  Section~\ref{sec:srpt_fsp}. We therefore define \emph{late} jobs in
  FSPE as those whose execution is completed in the virtual system but
  not yet in the real one and we notice that, analogously to SRPTE, also in
  FSPE late jobs can never be preempted by new ones, and they block
  the system until they are all completed.

\section{Our Solution}
\label{sec:solution}


Now that we have identified the issue with existing size-based
scheduling policies, we propose a strategy to avoid it.  It is
possible to envision strategies that update job size estimations as
work progresses in an effort to reduce errors; such
solutions, however, increase the complexity both in designing systems
and in analyzing them. In fact, the effectiveness of such a solution
would depend non-trivially on the way size estimation errors evolve
as jobs progress: this is inextricably tied to the way estimators are
implemented and to the application use case.
We propose, instead, a solution that requires no additional job
size estimation, based on the intuition that \emph{late jobs should
  not prevent executing other ones}. This goal is achievable
with simple modifications to preemptive size-based scheduling
disciplines such as SRPT and FSP; the key property is that the
scheduler takes corrective actions when one or more jobs are
\emph{late}, guaranteeing that newly arrived small jobs will execute soon
even when very large late jobs are running.

We conclude this section by showing our proposal, PSBS; it implements
this idea while being efficient ($O(\log n)$ complexity) and allowing
the usage of different weights to differentiate jobs. Our experimental
results show that PSBS achieves almost optimal mean sojourn times for
a large variety of workloads, suggesting that more complex solutions
involving re-estimations are unlikely to be very beneficial in many
practical cases.

\subsection{Using PS and LAS for Late Jobs}
\label{sec:unweighted-proposal}

From our analysis of Section~\ref{sec:under_over}, we understand that
current size-based schedulers behave problematically when one or more
jobs become late. Fortunately, it is possible to understand if jobs
are late from the internal state of the scheduler: in SRPT, a job is
late if its remaining estimated size is less than or equal to zero; in FSP, a job is late
if it is completed in the virtual time but not in the real time.

%


As outlined above, approaches that involve job size re-estimation are
difficult to design and evaluate, especially from the point of view of this work, where 
we do not make any assumption on the job size estimators; our
approach, therefore, requires only \emph{one} size estimation per job.

The key idea
of our proposal is that late jobs
should not monopolize the system resources. The solution is to modify the 
scheduler such that it provides service to a set of jobs, which we call
\emph{eligible jobs}, rather than a single job at a time. In particular,
we consider the following jobs as eligible when at least one job is late:
\begin{inparaitem}[]
\item for our amended version of SRPTE, all the
  late jobs, plus the non-late job with the highest-priority;
\item for our amended version of FSPE, \emph{only the
  late jobs}.
\end{inparaitem}

The two cases differ because, in
SRPTE, jobs only become late while they are being served since
remaining processing time decreases only for them;
therefore, non-late jobs need a chance to be served. We serve only one
non-late job to minimize unnecessary deviations from SRPTE. In FSPE,
conversely, jobs become late depending on the simulated behavior of
the virtual time, independently from which jobs are served in the real
time. 


We take into account two choices for scheduling eligible jobs: PS and
LAS (see Section~\ref{sec:related-ssq}). PS divides resources evenly
between all jobs, while LAS divides resources evenly
between the job(s) that received the least amount of service until the
current time.

The alternatives proposed so far lead to four scheduling policies that
we evaluate experimentally in Section~\ref{sec:experimental_results}:
\begin{enumerate}
\item \emph{SRPTE+PS.} Behaving as SRPTE as long as no jobs are late,
  switching to PS between all late jobs and the highest-priority
  non-late job;
\item \emph{SRPTE+LAS.} As above, but using LAS instead of PS;
\item \emph{FSPE+PS.} Behaving as FSPE as long as no jobs are late,
  switching to PS between all late jobs;
\item \emph{FSPE+LAS.} As above, but using LAS instead of PS.
\end{enumerate}

We point out that, in the absence of errors or just of size
underestimations, jobs are guaranteed to be never late; this means
that in such cases these scheduling policies will be equivalent to
SRPT(E) and FSP(E), respectively. For a more precise description, we point
the interested reader to their implementation in our
simulator.\footnote{\url{https://github.com/bigfootproject/schedsim/blob/4745b4b581029c4f9cbbb791f43386d32d0ef8f6/schedulers.py}}

\subsection{PSBS}
\label{sec:psbs}

In Section~\ref{sec:mst_against_ps} we show how the scheduling
protocols we propose outperform, in most cases, both existing
size-based scheduling policies and size-oblivious ones such as PS and
LAS. Betweeen the scheduling protocols just introduced, we point out
that FSPE+PS is the only one that guarantees to \emph{avoid
  starvation}: every job will eventually complete in the
virtual time, and therefore will be scheduled in a PS
fashion. Conversely, both SRPTE and LAS can starve large jobs if
smaller ones are continuously submitted.
Due to this property and to the good performance we observe in the
experiments of Section~\ref{sec:comparing_proposals}, we consider
FSPE+PS a desirable policy. It has, however, a few
shortcomings: first, it does not handle
weights to differentiate job priorities; 
second, its implementation is inefficient, requiring
$O(n)$ computation where $n$ is the number of jobs running in the emulated
system. Here, we propose PSBS, a generalization of
FSPE+PS which solves these problems, both allowing different
job weights and having an efficient $O(\log n)$ implementation.

\subsubsection{Job Weights}

Neither FSP nor PS support job differentiation through job weights. 
In particular, FSP
schedules jobs based on their completion time in a virtual time that
simulates an environment using PS, which treats all running
jobs equally.

To differentiate jobs in PSBS, we use Discriminatory Processor
Sharing (DPS)~\cite{aalto2007beyond} in the place of PS,
\emph{both in the virtual time and in the scheduling for late
  jobs}. DPS is a generalization of PS whereby each job is given a
weight, and resources are shared between processes proportionally to
their weight. By assigning a different weight to jobs, 
we can therefore prioritize important
jobs. When all weights are the same, DPS is equivalent to PS;
when each job has the same weight 
PSBS is equivalent to FSPE+PS.

Our handling of weights follows classic algorithms like
Weighted Fair Queuing (WFQ) and Weighted Round Robin (WRR),
accelerating aging (\ie, the decrease of
virtual size) proportionally to weight, so that jobs with higher
weight are scheduled earlier.
Our result from Section~\ref{sec:analytic} guarantees that
PSBS dominates DPS if job sizes are known exactly, while at the same time
being an \emph{online} scheduler, implemented without any knowledge of jobs
released in the future.

\subsubsection{Implementation}

FSP emulates a virtual system running a processor sharing (PS) discipline and keeps track of its job completion order; FSP then schedules one job at a time following that order. Whenever a new job arrives, FSP needs to update the remaining size of \emph{each} job in the emulated system to compute the new virtual finish times and the corresponding job completion order. Existing implementations of FSP \cite{Friedman2003,dell2014revisiting} have $O(n)$ complexity due to the job virtual remaining size update at each arrival.

In our implementation, we reduce the complexity of the update procedure. Before showing the details of the algorithm, we introduce an example to help understand our solution. Consider three jobs ($J_1$, $J_2$ and $J_3$) with sizes $s_1 = 10$, $s_2 = 5$ and $s_3 = 2$ respectively, weights $w_1 = w_2 = w_3 = 1$, which arrive at times $t = 0$, $t = 3$ and $t = 5$ respectively. Fig.~\ref{fig:virtual_finish_time} shows the evolution of the virtual emulated system, \ie, how the remaining virtual size decreases in the virtual time. For instance, when job $J_3$ arrives, since it will complete in the virtual time before jobs $J_1$ and $J_2$, it will be executed immediately in the real system (job $J_2$ will be preempted). To compute the completion time, it is possible to calculate the exact virtual remaining size of the jobs currently in the system.

\begin{figure}[!t]
  \centering
  \includegraphics[width=\columnwidth]{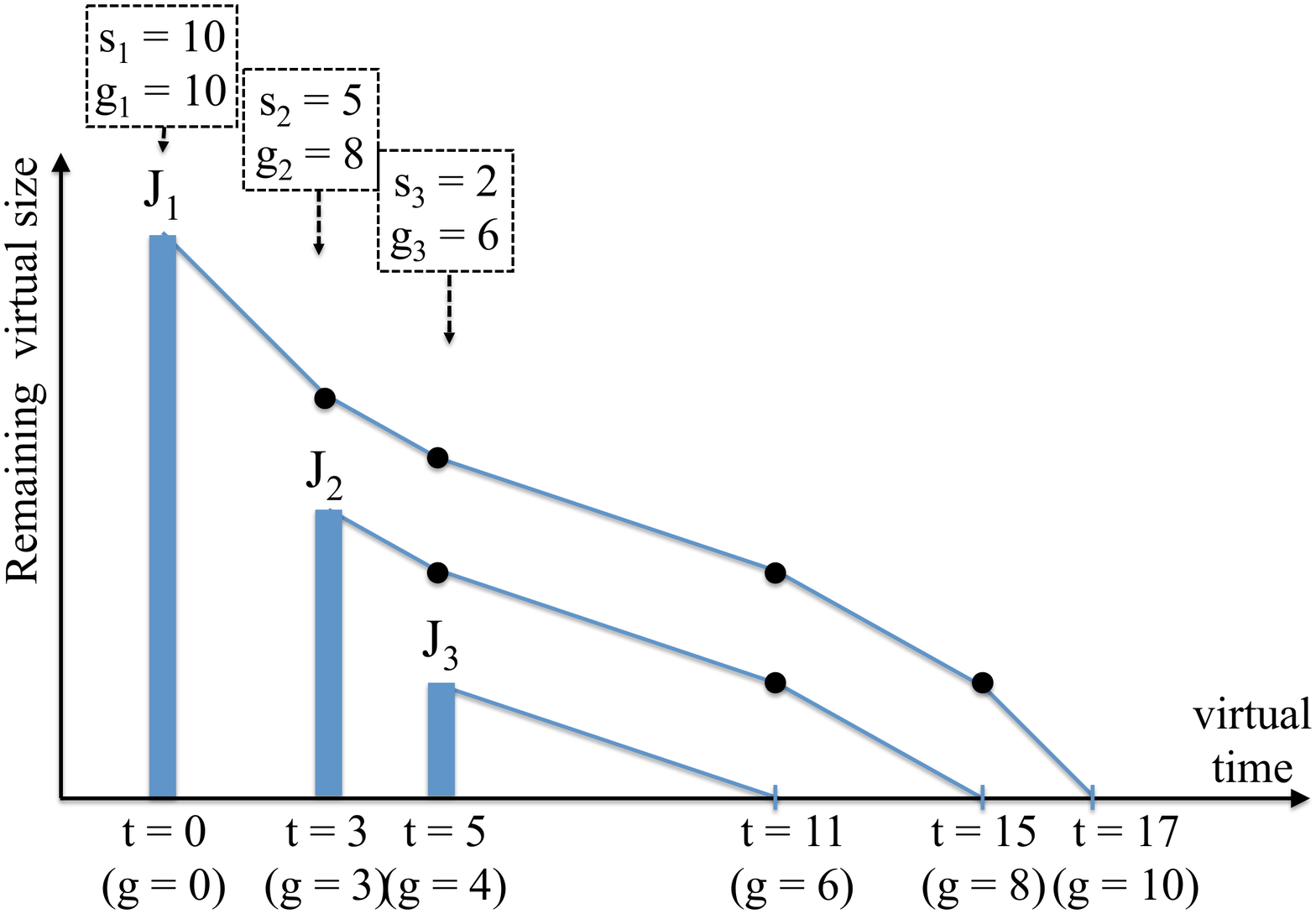}
  \caption{Example of virtual time $t$ and virtual lag $g$, used for sorting jobs.}
  \label{fig:virtual_finish_time}
\end{figure}

We instead introduce a new variable, which we call \emph{virtual lag} $g$. The key idea is that we store, for each job $i$, a job virtual lag $g_i$ so that \emph{$i$ completes in the virtual time when the virtual lag $g=g_i$}. We fulfill this property by updating $g$ at a rate that depends on the number of jobs in the system: for each time unit in the virtual time, $g$ increases by $1/w_v$, where $w_v$ is the sum of the weights $w_i$ of each job $i$ running in the virtual emulated system. 

Given a job $i$ with weight $w_i$ and size $s_i$ that arrives at the system when the virtual lag $g$ has a value $g = x$, the job virtual lag is given by $g_i=x+s_i/w_i$. The job virtual lag $g_i$ is computed just once, when the job arrives, and it does not need to be updated when other jobs arrive. In fact, only the global virtual lag $g$ needs to be updated according to the number of jobs in the system. Fig.~\ref{fig:virtual_finish_time} shows the job virtual lag computed when each jobs arrives (value of $g_i$ below $s_i$) and the value of the global virtual lag $g$ (below the virtual time $t$). Indeed, each job $i$ completes when $g = g_i$, but the only variable we update at each job arrival is $g$, leaving untouched the values $g_i$ of the job virtual lags. For instance, when job $J_3$ arrives, the virtual lag $g$ has value 4, therefore the job virtual lag will be $g_3 = 4 + 2$ (since the size and the weights of job $J_3$ are $s_3 = 2$ and $w_3 = 1$). It takes 6 time units (in the virtual time) to complete job $J_3$, which corresponds to 2 time units in the virtual lag. 

It is simple to show that, given any positive value for $s_i$ and $w_i$, the order at which jobs complete in the virtual time and in the virtual lag is exactly the same. Therefore, at each job arrival, it is sufficient to update the global virtual lag $g$, compute the job virtual lag $g_i$ and store the object in a priority queue, where the order is kept according to the values of $g_i$. The overall complexity is dominated by the maintenance of the priority queue, which is $O(\log n)$, since it is not necessary to update the virtual remaining size of all jobs in the system to compute the completion order.

The implementation of our solution, 
shown in Algorithm~\ref{alg:psbs}, follows the
nomenclature used in the original description of
FSP~\cite[Section~4.4]{Friedman2003}. We remark that, in the
absence of errors and when all job weights are the same, PSBS is
equivalent to FSP: therefore, our implementation of PSBS is also the
first $O(\log n)$ implementation of FSP.

\begin{algorithm}[!t]

  \tcc{Set up the scheduler state.\newline $\bv$ and $\be$ contain
    $(i, g_i, w_i)$ tuples: $i$ is the job id, $g_i$ is the value of
    $g$ at which the job completes in the virtual time and $w_i$ is
    the weight. They are sorted by $g_i$.}

  \Fn{Init}{
    \newcommand{\commentalign}[1]{\makebox[.27\linewidth][l]{#1}}

    \commentalign{$g \leftarrow 0$}
    \tcp{virtual lag (see text)}

    \commentalign{$t \leftarrow 0$}
    \tcp{virtual time}

    \tcp{virtual time queue}
    $\bv \leftarrow$ empty binary min-heap\;

    \tcp{``early'' jobs completed in real time}
    $\be \leftarrow$ empty binary min-heap\;

    \tcp{mapping from job ids of late jobs to their weight}
    $\bl \leftarrow$ empty hashtable\;

    \commentalign{$w_\bl \leftarrow 0$}
    \tcp{$\sum w_i$ for each late job $i$}

    \commentalign{$w_v \leftarrow 0$}
    \tcp{$\sum w_i \forall i$ running in virtual time}

  }

  \;

  \Fn{NextVirtualCompletionTime}{
    \If{$\bv$ and/or $\be$ are not empty}{
      $\hat g \leftarrow \min\{\text{first $g_i$ in $\bv$}, \text{first
        $g_i$ in $\be$\}}$\;
      \Return $t + w_v(\hat g - g)$
    }\lElse{
      \Return $\emptyset$
    }
  }
  
  \;

  \Fn{UpdateVirtualTime($\hat t$)}{
    \lIf{$w_v > 0$}{
      $g \leftarrow g + (\hat t - t)/w_v$
    }
    $t \leftarrow \hat t$
  }

  \;

  \Fn{VirtualJobCompletion($\hat t$)} {
    UpdateVirtualTime($\hat t$)\;
    \If{$\text{first $g_i$ in $\bv$} \leq g$}{
      $(i, \_, w_i) \leftarrow \text{pop($\bv$)}$\;
      $\bl[i] \leftarrow w_i$\;
      $w_\bl \leftarrow w_\bl + w_i$
    }\Else(\tcp*[h]{the virtual job that completes is in $\be$}){
      $(\_, \_, w_i) \leftarrow \text{pop($\be$)}$
    }
    $w_v \leftarrow w_v - w_i$
  }

  \;
  
  \Fn{RealJobCompletion($i$)}{
    \If(\tcp*[h]{we were scheduling late jobs}){$\bl$ is not empty}{
      $w_i \leftarrow pop(\bl[i])$\;
      $w_\bl \leftarrow w_\bl - w_i$\;
    }\Else(\tcp*[h]{we were scheduling the first job in $\bv$}){
      push pop($\bv$) into $\be$
    }
  }

  \;

  \Fn{JobArrival($\hat t, i, s_i, w_i$)}{
    UpdateVirtualTime($\hat t$)\;
    push $(i, g + s_i/w_i, w_i)$ into $\bv$\;
    $w_v \leftarrow w_v + w_i$\;
  }

  \;

  \Fn{ProcessJob}{
    \lIf{$\bl$ is not empty}{
      \Return $\left\{(i, w_i/w_\bl): (i, w_i) \in \bl\right\}$
    }\lElseIf{$\bv$ is not empty}{
      \Return $\{(\text{first job id of $\bv$}, 1)\}$
    }\lElse{
      \Return{$\emptyset$}
    }
  }

  \caption{PSBS.}
  \label{alg:psbs}
\end{algorithm}

Computation is triggered by three events: if a job $i$ of weight $w_i$
and estimated size $s_i$ arrives at time $\hat t$, JobArrival($\hat t,
i, s_i, w_i$) is called; when a job $i$ completes,
RealJobCompletion($i$) is called; finally, when a job completes in
virtual time at time $\hat t$, VirtualJobCompletion($\hat t$) is
called (NextVirtualCompletionTime is used to discover when to call
VirtualJobCompletion). After each event, ProcessJob is
called to determine the new set of scheduled jobs: its output is a set
of $(j, s)$ pairs where $j$ is the job identifier and $s$ is the
fraction of system resources allocated to it.

As auxiliary data structures, we keep two priority queues, $\bv$ and $\be$. 
$\bv$ stores jobs that are
running both in the real time and in the virtual time, while $\be$
stores ``early'' jobs that are still running in the virtual time but
are completed in the real time. For each job $i$, we store in $\bv$ or
$\be$ an immutable tuple $(i, g_i, w_i)$ containing respectively the
job id, the virtual lag $g_i$ and the weight. We use binary min-heaps to
represent $\bv$ and $\be$, using the $g_i$ values as ordering key:
binary heaps are efficient data structures offering worst-case $O(\log
n)$ ``push'' and ``pop'' operations, $O(1)$ lookup of the first value
and eassentially optimal memory efficiency, by virtue of being an implicit data
structure requiring no pointers~\cite{Cormen:2001:IA:580470}. In
addition, the push operation has of $O(1)$ complexity on
average~\cite{porter1975random}.
The state of the scheduler is completed by a mapping $\bl$ from the
identifiers of late jobs to their weight, a counter $t$ representing
the virtual time, and two variables $w_v$ and
$w_\bl$ representing the sum of weights for jobs that are
respectively active in the virtual time and late. Some additional
bookkeeping, not included here for simplicity, would be needed to handle
jobs that complete even when they are not scheduled (\eg, because of
error conditions or after being killed): we refer the interested
reader to the implementation in our
simulator.\footnote{\url{https://github.com/bigfootproject/schedsim/blob/4745b4b581029c4f9cbbb791f43386d32d0ef8f6/schedulers.py}}
Additional details can be found in the supplemental material.

\paragraph*{Complexity Analysis}

We consider here average complexity due to the worst-case $O(n)$
complexity of hashtable operations.\footnote{A denial-of-service
  attack on hashtables has been designed by forging keys
  to obtain collisions~\cite{klink2011effective}. This attack is
  defeated in modern implementations by salting keys before hashing.}
It is trivial to see that
NextVirtualCompletionTime and UpdateVirtualTime have $O(1)$
complexity. Since inserting elements in hashtables has $O(1)$ average
complexity, the cost of VirtualJobCompletion is dominated by the pop
operations on $\bv$ and $\be$: both of them are bound by $O(\log n)$,
where $n$ is the number of jobs in the system. Removing an element
from a hashtable has $O(1)$ average cost, so the cost of
RealJobCompletion is dominated by the pop on $\bv$, which has again
$O(\log n)$ complexity. JobArrival has $O(1)$ average complexity
(remember that pushing elements on a binary heap is $O(1)$ on
average).

The ProcessJob procedure, when $\bl$ is not empty, has $O(|\bl|)$
complexity because the output itself has size $\bl$.
This is however very unlikely to be a limitation in
practical cases, since real-world implementations of schedulers
allocate resources one by one in discrete slots: schedulers such as PS
or DPS are abstractions of mechanisms such as round-robin or max-min
fair schedulers, which can be implemented efficiently; a real-world
implementation of PSBS would adopt similar strategy to mimick the DPS-like
resource sharing when $\bl$ is not empty. We also note that, when
there are no job size estimation errors and PSBS is used to implement
FSP, $\bl$ is guaranteed to always be empty and therefore ProcessJob
will have $O(1)$ complexity.

As we have seen, with the exclusion of ProcessJob as discussed above,
all the procedures of the scheduler have at most $O(\log n)$
computational complexity. Coupled $O(\log n)$
operations having low constant factors because they are implemented on
binary heaps, which are very efficient data structures, we believe
that these performance guarantees are sufficient for a very large set
of practical situations: for example, CFS -- the current Linux
scheduler -- has $O(\log n)$ complexity since it uses a tree
structure~\cite{cfs}.

\section{Evaluation Methodology}
\label{sec:simulator}

Understanding size-based scheduling when there are estimation errors
is not a simple task; analytical studies have been performed only with
strong assumptions such as bounded error~\cite{wierman2008scheduling}.
Moreover, to
the best of our knowledge, the only analytical result known for FSP
(without estimation errors) is its dominance over PS, making
analytical comparisons between SRPTE-based and FSPE-based scheduling
policies even more difficult.

For these reasons, we evaluate our proposals
through simulation. The simulative approach is extremely flexible,
allowing to take into account several parameters -- distribution of the
arrival times, of the job sizes, of the errors. Previous simulative
studies (\eg,~\cite{lu2004size}) have focused on a subset of
these parameters, and in some cases they have used real traces. In our
work, we developed a tool that is able to both reproduce real traces
and generate synthetic ones. Moreover, thanks to the efficiency of the
implementation, we were able to run an extensive evaluation campaign,
exploring a large parameter space. For these
reasons, we are able to provide a broad view of the applicability of
size-based scheduling policies, and show the benefits and the
robustness of our solution with respect to the existing ones.

\subsection{Scheduling Policies Under Evaluation}
\label{sec:schedpolicies}

In this work, we take into account different scheduling policies, both size-based and size-oblivious. For the size-based disciplines, we consider SRPT as a reference for its optimality with respect to the MST. When introducing the errors, we evaluate SRPTE, FSPE 
and our proposals described in Section~\ref{sec:solution}.

As size-oblivious policies, we have implemented the
\emph{First In, First Out} (FIFO) and \emph{Processor Sharing} (PS)
disciplines, along with DPS, the generalization of PS with
weights~\cite{kleinrock1976queueing}.
These policies are the default disciplines used in many
scheduling systems -- \eg, the default scheduler in
Hadoop~\cite{white2009hadoop} implements a FIFO policy, while Hadoop's
FAIR scheduler is inspired by PS; the Apache web server delegates
scheduling to the Linux kernel, which in turn implements a PS-like
strategy~\cite{schroeder2006web}. Since PS scheduling divides evenly
the resources among running jobs, it is generally considered as a
reference for its fairness (see the next section on the performance
metrics). Finally, we consider also the \emph{Least Attained Service}
(LAS)~\cite{rai2003analysis} policy. LAS scheduling is a preemptive
policy that gives service to the job that has received the least
service, sharing it equally in a PS mode in case of ties. LAS
scheduling has been designed considering the case of heavy-tailed job
size distributions, where a large percentage of the total work
performed in the system is due to few very large jobs, since it gives
higher priority to small jobs than what PS would do.

\subsection{Performance Metrics}
\label{sec:metrics}

We evaluate scheduling policies according to two main aspects:
\emph{mean sojourn time} (MST) and \emph{fairness}. 
Sojourn time is the time that passes between the
moment a job is submitted and when it completes; such a metric is
widely used in the scheduling literature.
The definition of fairness is more elusive: in his survey on the
topic, \citeN{wierman2011fairness} affirms that \textit{``fairness is an amorphous concept
that is nearly impossible to define in a universal
way''}. When the job size distribution is
skewed, it is intuitively unfair to expect similar sojourn times
between very small jobs and much larger ones; a common approach is to
consider \emph{slowdown}, \ie the ratio between a job's sojourn time
and its size, according to the intuition that the waiting time for a
job should be somewhat proportional to its size. In this work we focus
on the per-job slowdown, to check that as few jobs as
possible experience ``unfair'' high slowdown values;
moreover, in accordance with Wierman's definition~\cite{wierman2007fairness},
we also evaluate \emph{conditional slowdown},
which evaluates the expected slowdown given a job size, verifying whether jobs
of a particular size experience an ``unfair'' high expected slowdown value.

\subsection{Parameter Settings}
\label{sec:trace_generation}

\begin{table}[!t]
    \centering
    \begin{tabular}{|l|l|r|}
      \hline
      Parameter & Explanation & Default\\
      \hline
      \hline
      sigma & $\sigma$ in the log-normal error distribution & 0.5\\
      shape & shape for Weibull job size distribution & 0.25\\
      timeshape & shape for Weibull inter-arrival time & 1\\
      njobs & number of jobs in a workload & 10,000\\
      load & system load & 0.9\\
      \hline
    \end{tabular}
    \caption{Simulation parameters.}
  \label{table:parameters}
  
\end{table}

We empirically evaluate scheduling policies in a wide spectrum of
cases. Table~\ref{table:parameters} synthetizes the input parameters
of our simulator; they are discussed in the following.

\paragraph*{Job Size Distribution}
Job sizes are generated according to a
Weibull distribution, which allows us to evaluate both heavy-tailed
and light-tailed job size distributions. The
\emph{\textbf{shape}} parameter allows to interpolate between
heavy-tailed distributions (shape $<1$), the exponential distribution
(shape$=1$), the Raleigh distribution ($\text{shape}=2$) and
light-tailed distributions centered around the `1' value
($\text{shape}>2$). We set the \emph{scale} parameter of the
distribution to ensure that its mean is 1.

Since scheduling problems have been generally analyzed on heavy-tailed
workloads with job sizes using distributions such as Pareto, we
consider a default heavy-tailed case of $\text{shape}=0.25$. In our
experiments, we vary the shape parameter between a very skewed
distribution with $\text{shape}=0.125$ and a light-tailed
distribution with $\text{shape}=4$.

\paragraph*{Size Error Distribution}
We consider log-normally distributed
errors. A job having size $s$ will be estimated as $\hat{s}=sX$,
where $X$ is a random variable with distribution
\begin{equation}
\operatorname{Log-\mathcal{N}}(0,\sigma^{2}).
\label{eq:lognormal}
\end{equation}

This choice satisfies two properties: first, since error is
multiplicative, the absolute error $\hat s - s$ is proportional to the
job size $s$; second, under-estimation and over-estimation are equally
likely, and for any $\sigma$ and any factor $k > 1$ the (non-zero)
probability of under-estimating $\hat s\leq\frac s k$ is the same of
over-estimating $\hat s\geq ks$.  This choice also is substanciated by
empirical results: in our implementation of the HFSP scheduler for
Hadoop~\cite{pastorelli2013hfsp}, we found that the empirical error
distribution was indeed fitting a log-normal distribution.

The \emph{\textbf{sigma}} parameter
controls $\sigma$ in Equation~\ref{eq:lognormal}, with a default --
used if no other information is given -- of 0.5; with this value, the
median factor $k$ reflecting relative error is 1.40. In our
experiments, we let sigma vary between 0.125 (median $k$ is 1.088) and
4 (median $k$ is 14.85). 

It is possible to compute the correlation between the estimated and real
size as $\sigma$ varies. In particular, when sigma is equal to 0.5, 1.0, 2.0 and 4.0,  
the correlation coefficient is equal to 0.9, 0.6, 0.15 and 0.05 respectively.

The mean of this distribution is always larger than 1,
  and, as sigma grows, the system is biased towards
  overestimating the aggregate size of several jobs, limiting the
  underestimation problems that our proposals are designed
  to solve. Even in this setting, the results in
  Section~\ref{sec:experimental_results} show that the improvements
  we obtain are still significant.

\paragraph*{Job Arrival Time Distribution}
For the job inter-arrival time
distribution, we use again a Weibull distribution for its flexibility to
model heavy-tailed, memoryless and light-tailed distributions.  We set
the default of its shape parameter (\emph{\textbf{timeshape}}) to 1,
corresponding to ``standard'' exponentially distributed arrivals. Also
here, timeshape varies between 0.125 (very bursty arrivals separated
by long intervals) and 4 (regular arrivals).

\paragraph*{Other Parameters}
The \emph{\textbf{load}} parameter is the mean arrival rate divided by
the mean service rate. As a default, we use 0.9 like \citeN{lu2004size};
in our
experiments we let it
vary between 0.5 and 0.999.
The number of jobs (\emph{\textbf{njobs}}) in each simulation round is
10,000. For each experiment, we perform at least 30
repetitions, and we compute the confidence interval for a confidence
level of 95\%. For very heavy-tailed job size distributions (shape
$\leq 0.25$), results are very variable and therefore, to
obtain stable averages, we performed hundreds and/or thousands of
experiment runs, at least until the confidence levels have reached the 5\% of
the estimated values.

\section{Experimental Results}
\label{sec:experimental_results}

We now proceed to an extensive report of our experimental findings. We
first provide a high-level view showing that our proposals outperform PS,
excepting only extreme cases of \emph{both} error \emph{and} job skew
(Section~\ref{sec:mst_against_ps}); we then proceed to a more in-depth
comparison of our proposals, to validate our choice of using FSPE+PS as
a base for PSBS (Section~\ref{sec:comparing_proposals}). We then evaluate the
performance of PSBS against existing schedulers, while varying the two
parameters that most influence scheduler performance: shape
(Section~\ref{sec:shape}) and sigma (Section~\ref{sec:sigma}). We
proceed to show that PSBS handles jobs fairly
(Section~\ref{sec:fairness}) and that job weights are handled
correctly (Section~\ref{sec:priority}); we conclude our analysis on
synthetic workloads by showing that our results hold even while
varying settings over the parameter space
(Section~\ref{sec:other_settings}). We conclude our analysis by
comparing PSBS to existing schedulers on real workloads extracted from
Hadoop logs and an HTTP cache (Section~\ref{sec:real_workloads}).

For the results shown in
the following, parameters whose values are not explicitly stated
take the default values in Table~\ref{table:parameters}. For
readability, we do not show the confidence intervals: 
for all the points, in fact, we have
performed a number of runs sufficiently high to obtain a confidence
interval smaller than 5\% of the estimated value. Where not otherwise
stated, all the $w_i$ parameters representing the weight of each job
$i$ have always been set to 1.

\subsection{Mean Sojourn Time Against PS}
\label{sec:mst_against_ps}

\begin{figure*}[!t]
  \centering
  \subfloat[SRPTE.]{
    \includegraphics[width=\threeplotwidth]{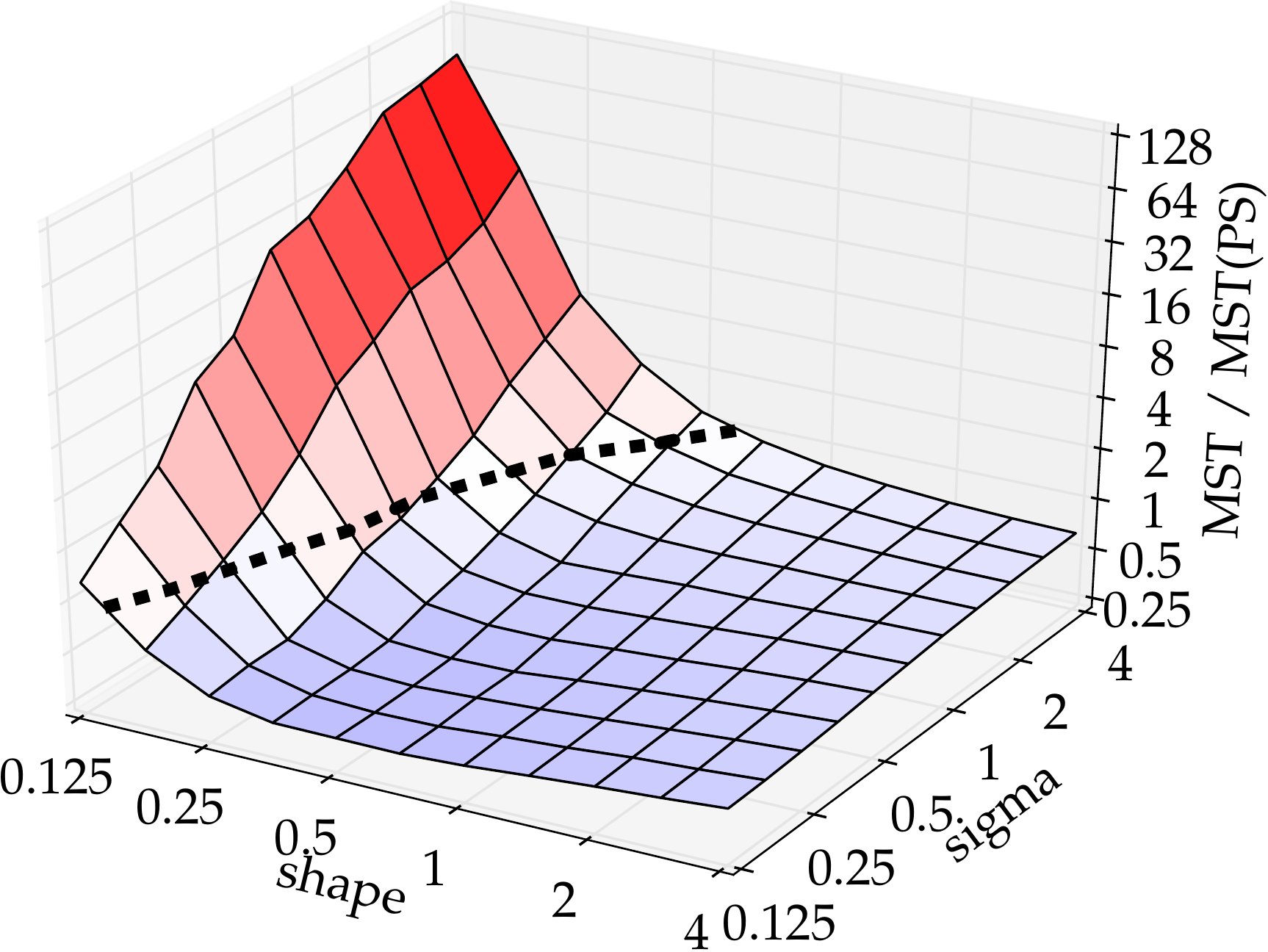}
    \label{fig:3d_srpte}
  }
  \subfloat[SRPTE+PS.]{
    \includegraphics[width=\threeplotwidth]{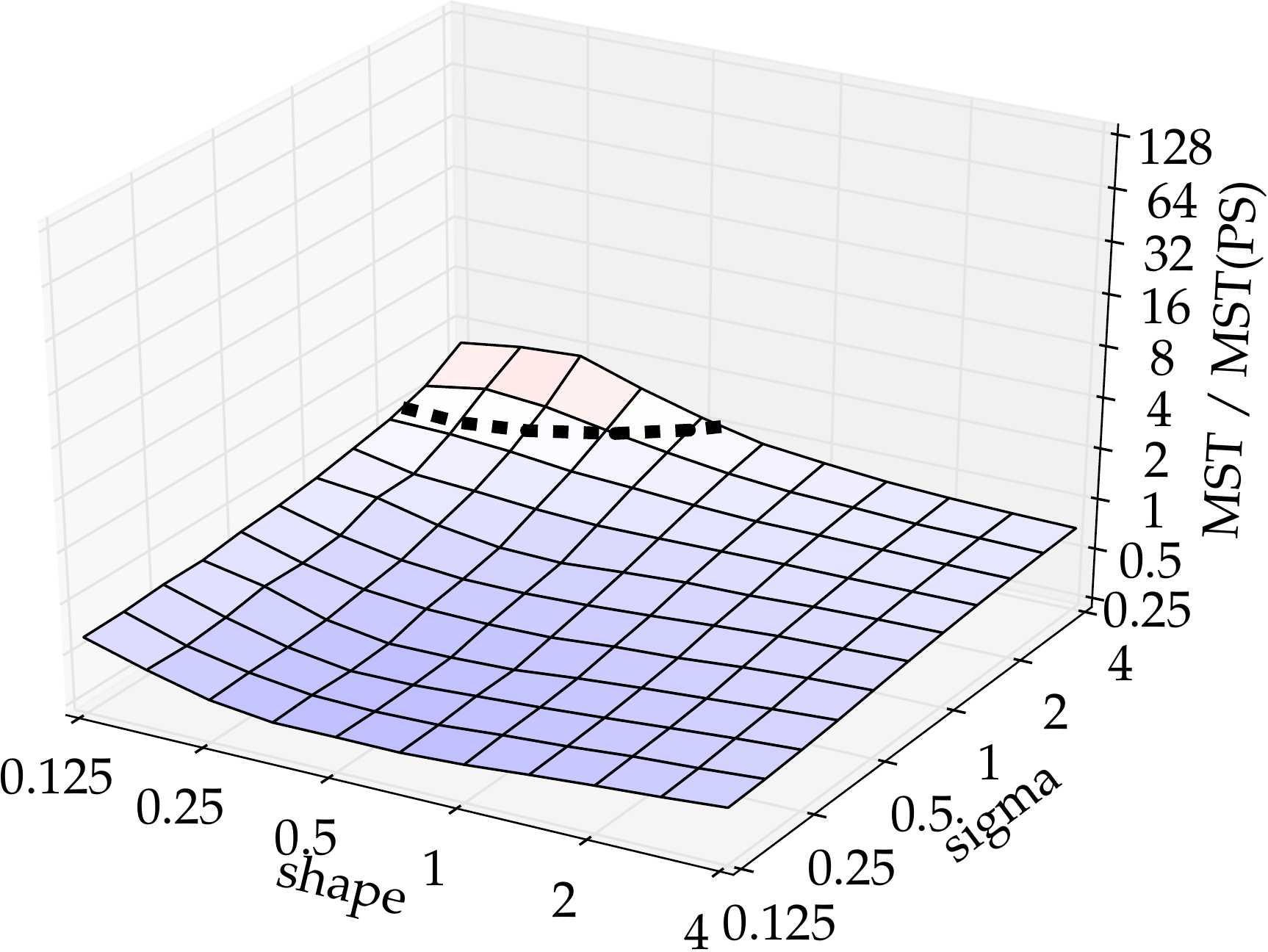}
    \label{fig:3d_srpte+ps}
  }
  \subfloat[SRPTE+LAS.]{
    \includegraphics[width=\threeplotwidth]{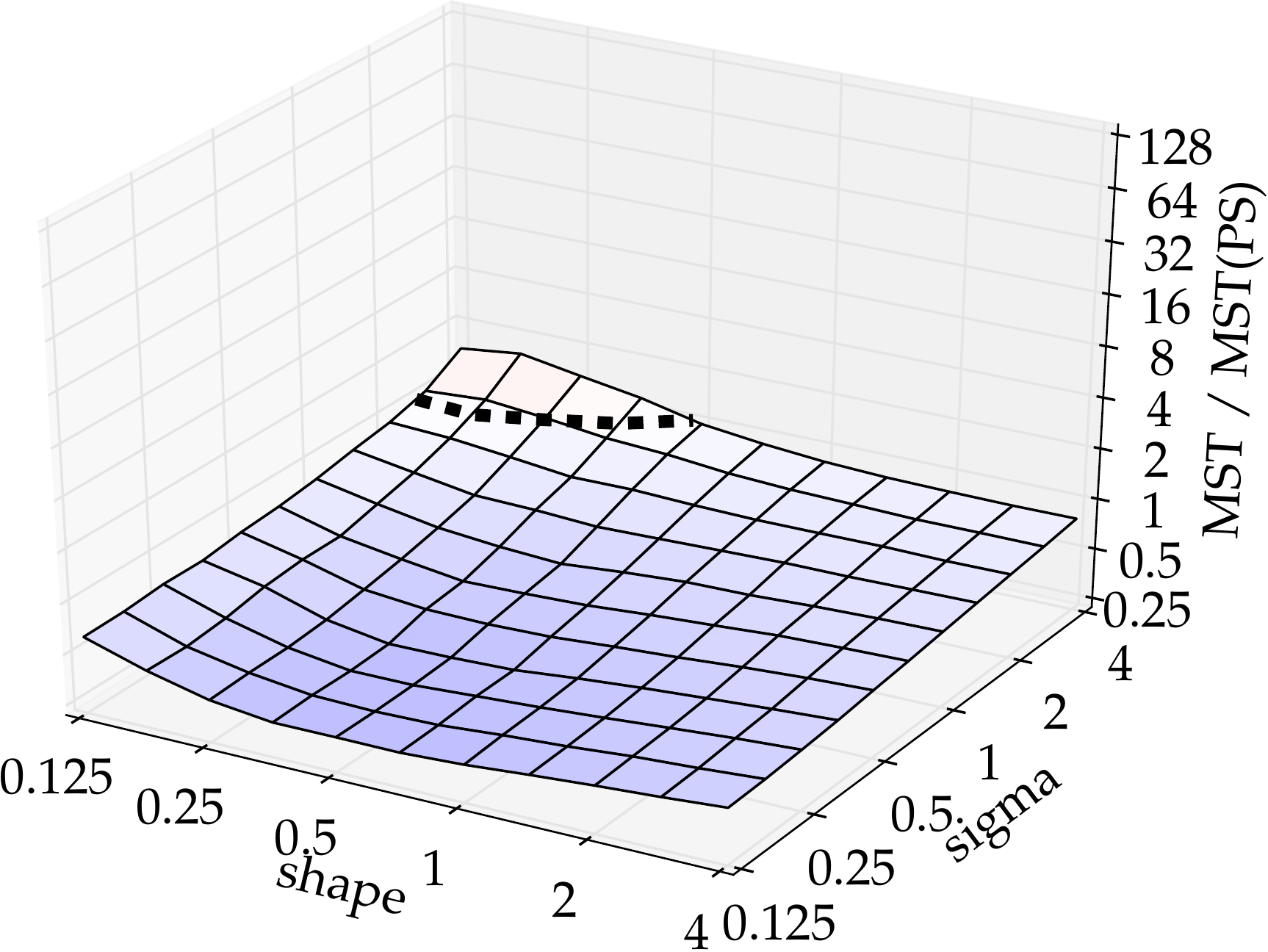}
    \label{fig:3d_srpte+las}
  }
  \\
  \subfloat[FSPE.]{
    \includegraphics[width=\threeplotwidth]{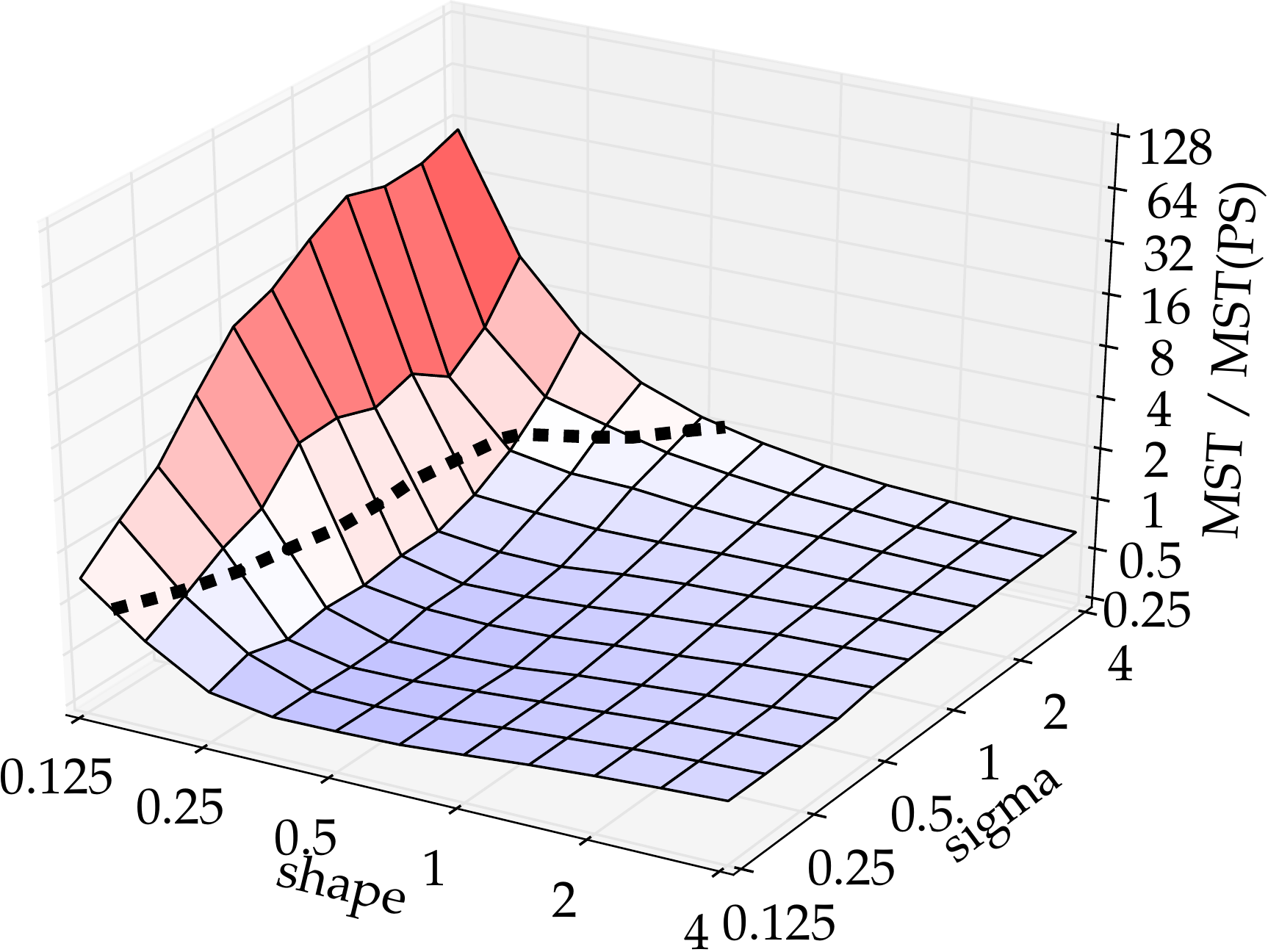}
    \label{fig:3d_fspe}
  }
  \subfloat[FSPE+PS.]{
    \includegraphics[width=\threeplotwidth]{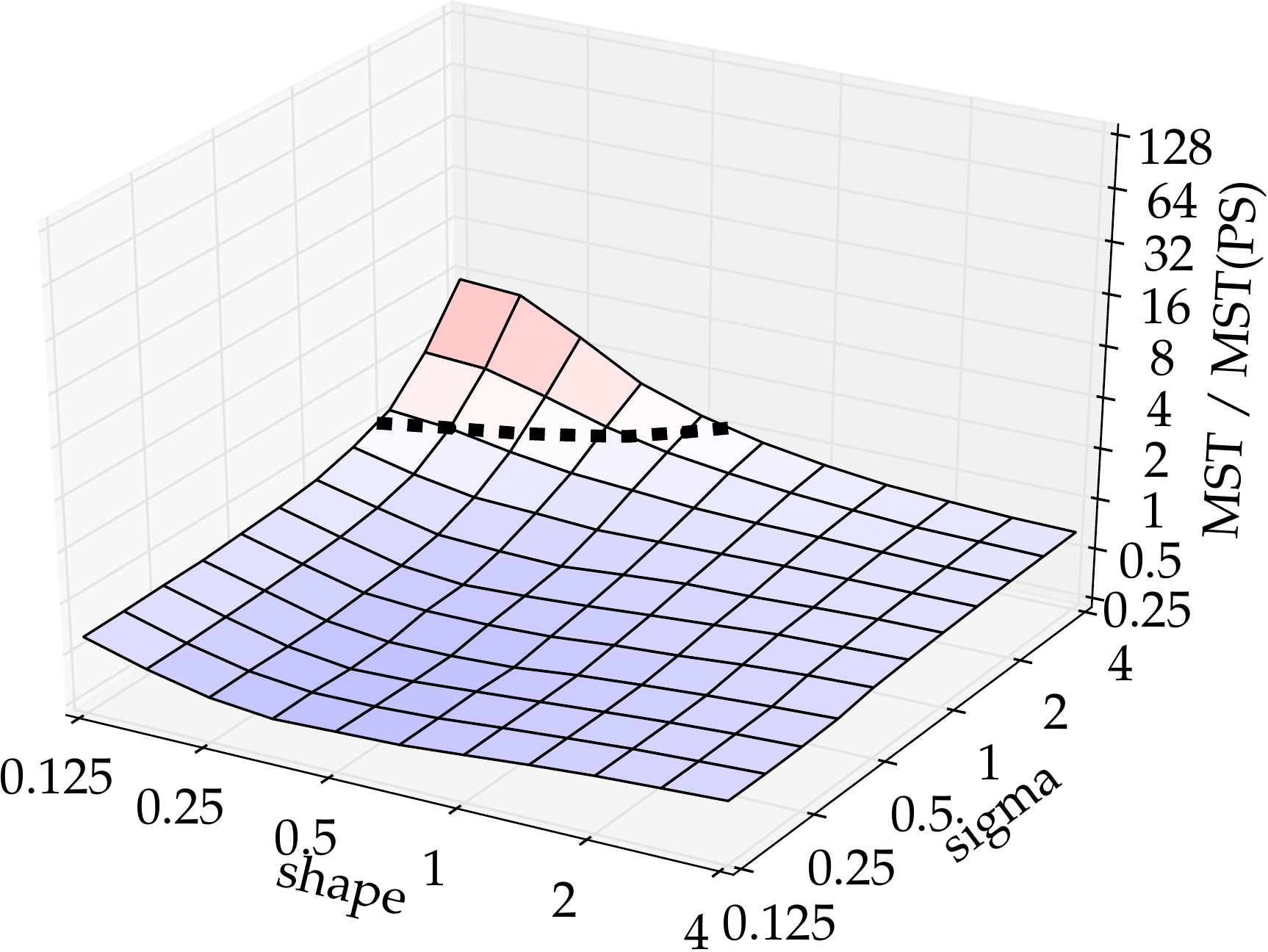}
    \label{fig:3d_fspe+ps}
  }
  \subfloat[FSPE+LAS.]{
    \includegraphics[width=\threeplotwidth]{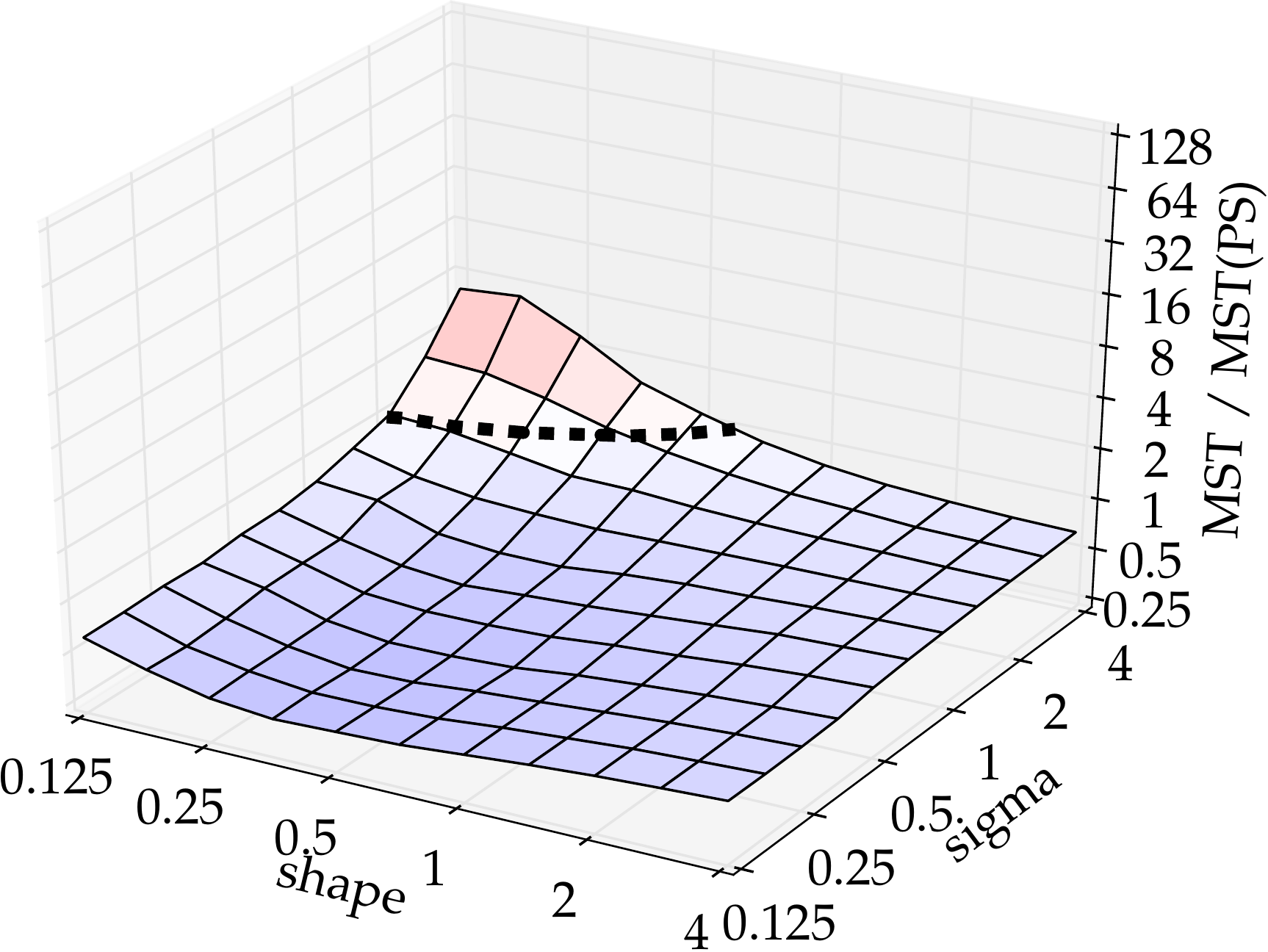}
    \label{fig:3d_fspe+las}
  }
  \caption{Mean sojourn time against PS: the dashed line is the
    boundary where MST is equivalent to that of PS. We recall that a low
    shape value is associated to high job size skew, while high sigma
    entails imprecise job size estimates.}
  \label{fig:3d}
\end{figure*}


We begin our analysis by comparing
the size-based scheduling policies, using PS as a baseline
because PS and its variants are the most widely used set of scheduling
policies in real systems. In Fig.~\ref{fig:3d} we plot the value of
the MST obtained using SRPTE, FSPE and the four alternatives
we propose in Section~\ref{sec:unweighted-proposal},
normalizing it against the MST of PS. We vary the sigma and shape
parameters influencing respectively job size distribution and error
rate; we will see that these two parameters are the ones that
influence performance the most. Values lower than one (below the
dashed line in the plot) represent regions where size-based schedulers
perform better than PS.

In accordance to intuition and to what is known from the literature,
we observe that the performance of size-based scheduling policies
depends on the accuracy of job size estimation: as sigma grows,
performance suffers. In addition, from Figures~\ref{fig:3d_srpte}
and~\ref{fig:3d_fspe}, we observe a new phenomenon: \emph{job
  size distribution impacts performance even more than size estimation
  error.} On the one hand, we notice that large areas of the plots
($\textrm{shape} > 0.5$) are almost insensitive to estimation errors;
on the other hand, we see that MST becomes very large as job size skew
grows ($\textrm{shape} < 0.25$). We attribute this latter phenomenon
to the fact that, as we highlight in Section~\ref{sec:errors}, late
jobs whose estimated remaining (virtual) size reaches zero are never
preempted. If a large job is under-estimated and becomes \emph{late} with
respect to its estimation, small jobs will have to wait for it to
finish in order to be served.

As we see in
Figures~\ref{fig:3d_srpte+ps},~\ref{fig:3d_srpte+las},~\ref{fig:3d_fspe+ps} and~\ref{fig:3d_fspe+las}, 
\emph{our proposals outperform
  PS in a large class of heavy-tailed workloads} where SRPTE and FSPE
suffer. The net result is that the size-based policies we propose are
outperformed by PS only in extreme cases where \emph{both} the job size
distribution is extremely skewed \emph{and} job size estimation is
very imprecise.

It may appear surprising that, when job size skew is not extreme,
size-based scheduling can outperform PS even when size estimation
is very imprecise: even a small correlation between job size
and its estimation can direct the scheduler towards choices that are
beneficial on aggregate. In fact, as we see more in detail in the
following (Section~\ref{sec:shape}), sub-optimal scheduling choices
become less penalized as the
job size skew diminishes.

\subsection{Comparing Our Proposals}
\label{sec:comparing_proposals}

\begin{figure*}[!t]
  \centering
  \subfloat[$\text{shape}=0.25$.]{
    \includegraphics[width=\threeplotwidth]{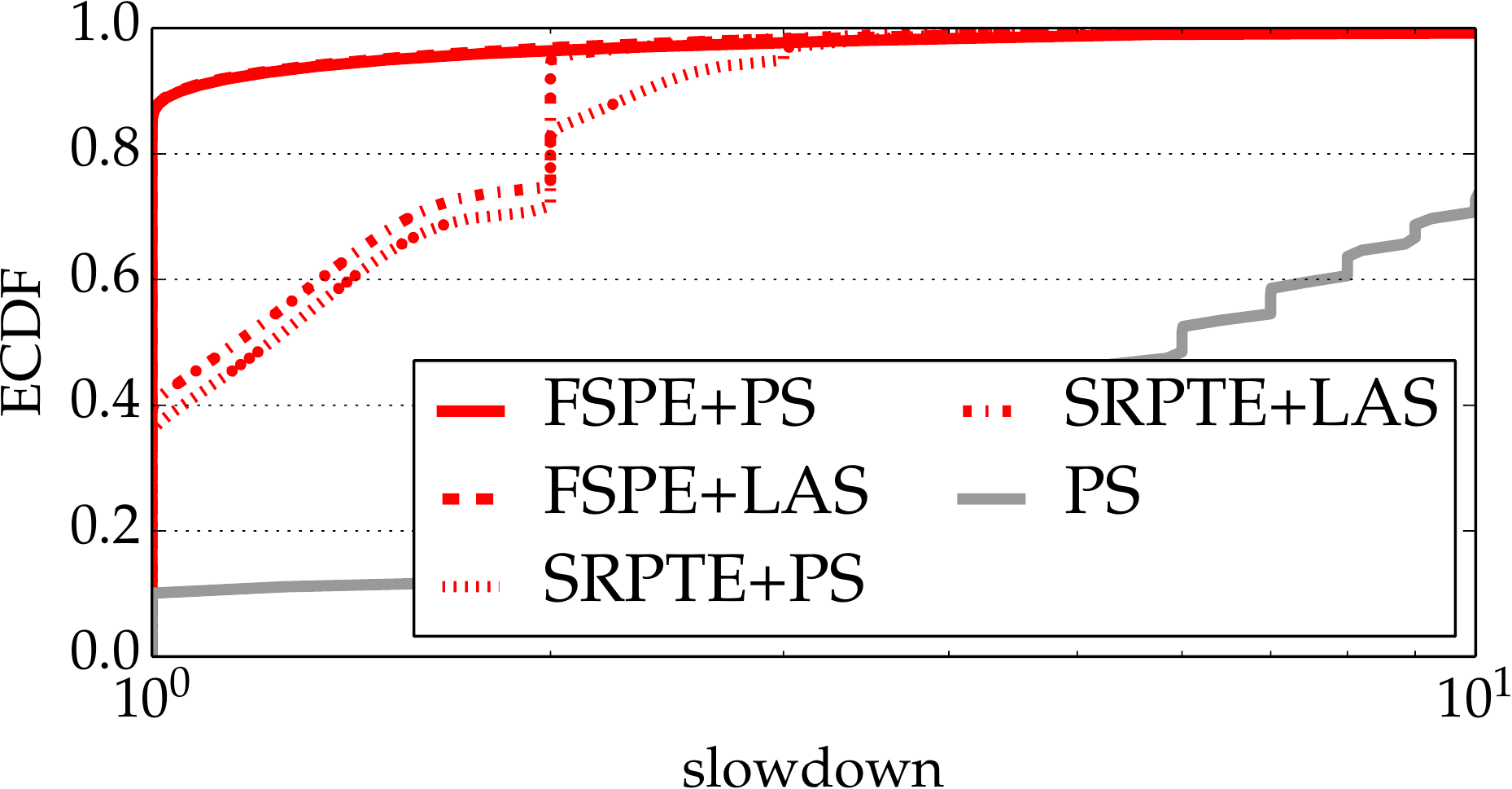}
    \label{fig:alt_025}
  }
  \subfloat[$\text{shape}=1$.]{
    \includegraphics[width=\threeplotwidth]{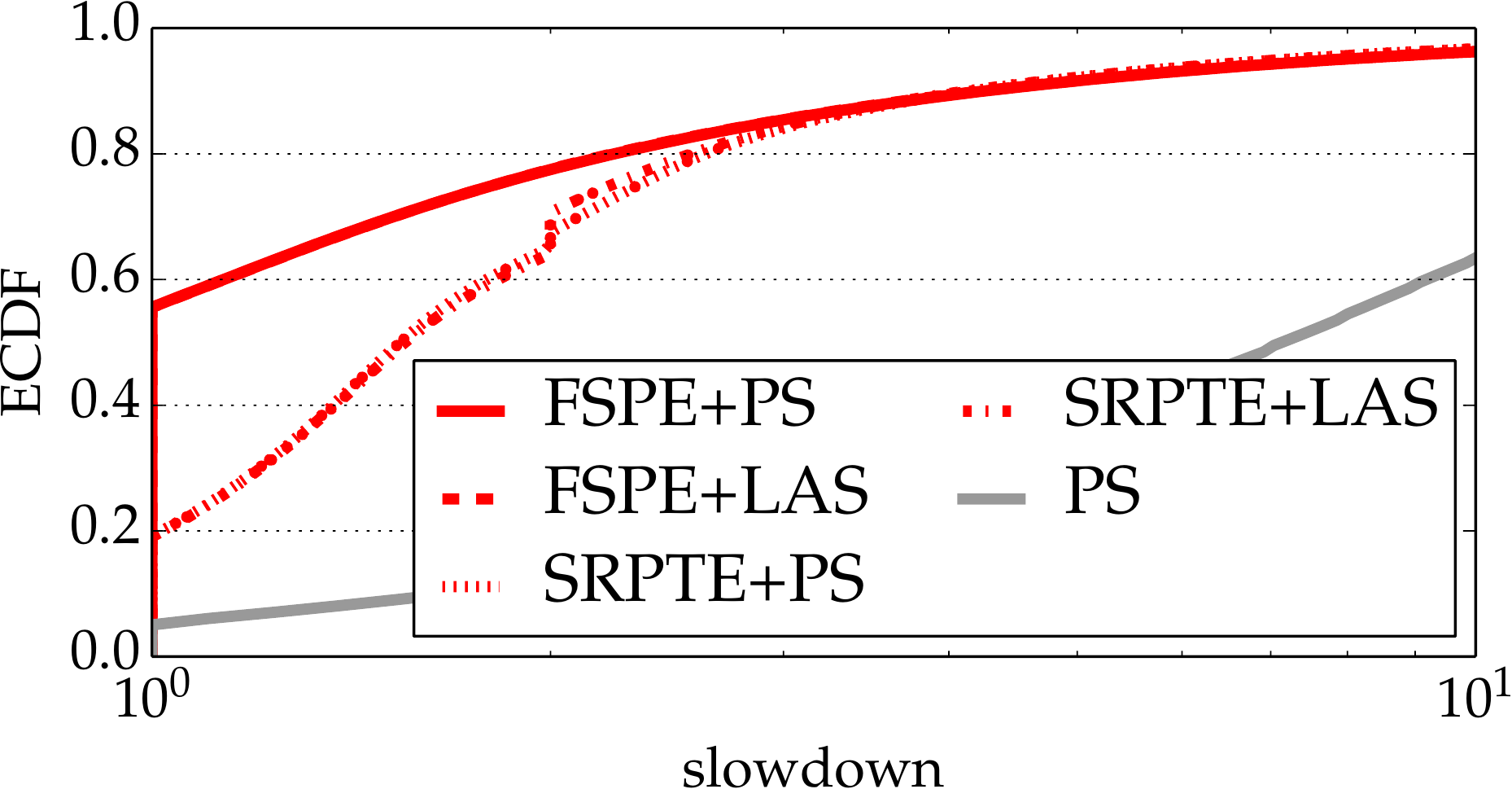}
  }
  \subfloat[$\text{shape}=4$.]{
    \includegraphics[width=\threeplotwidth]{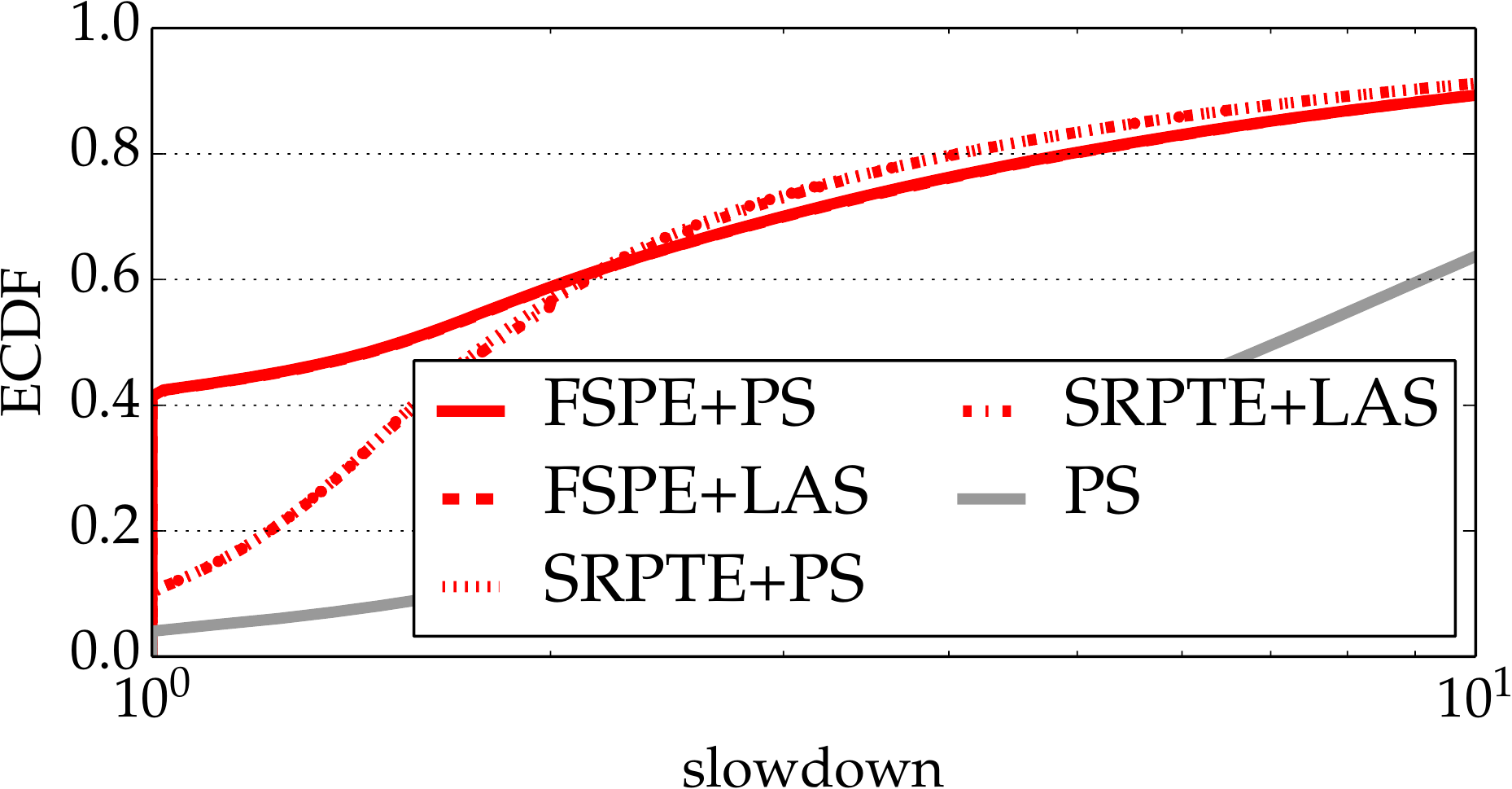}
  }
  \caption{Distribution of per-job slowdown. The two FSPE-based
    policies perform best, with negligible differences between them.}
  \label{fig:alt_slowdown}
\end{figure*}

How do the schedulers we proposed in
Section~\ref{sec:unweighted-proposal} compare? In
Fig.~\ref{fig:alt_slowdown} we examine the
empirical cumulative distribution function (ECDF) of the slowdown for
all jobs we simulate while varying the shape parameter (sigma
maintains its default value
of 0.5); we plot the results for PS as a reference
and observe that the staircase-like pattern observable in
Fig.~\ref{fig:alt_025} is a clustering around integer values
obtained if a small job gets submitted while $n$
larger ones are running.

We observe that, in general, our proposals pay off: for all values of
shape considered, the slowdown distribution of our proposals is well
lower than the one of PS. We also observe a difference between the
schedulers based on SRPTE and those based on FSPE: a noticeably larger
number of jobs experience an optimal slowdown of 1 when using a
scheduler based on FSPE. This is because, when using FSPE-based
scheduling policies, the number of jobs that are eligible for PS- or
LAS-based scheduling is higher: when late jobs exist, only they are eligible to be
scheduled, unlike what happens in SRPTE-based policies; as a
consequence, several small jobs suffer in SRPTE-based policies because
they are preempted too aggressively: as soon as they become late, even if they are the only late job in the system. This confirms the soundness of
the design policy we adopted in
Section~\ref{sec:unweighted-proposal}:
minimizing the number of eligible jobs for PS- or
LAS-based scheduling.
Fig.~\ref{fig:alt_slowdown}
shows that even allowing to schedule \emph{a single} non-late job can
hurt performance.

Since the number of late jobs is generally small, differences in
scheduling between FSPE+PS and FSPE+LAS are rare. This is confirmed by
noticing that the lines for the two schedulers in
Fig.~\ref{fig:alt_slowdown} are essentially analogous; we conclude
that FSPE+PS and FSPE+LAS have essentially analogous performance. This
fact and the property that FSPE+PS avoids starvation, as noted in
Section~\ref{sec:psbs}, motivated us to develop PSBS as a
generalization of FSPE+PS.

\subsection{Impact of Shape}
\label{sec:shape}

\begin{figure}[!t]
  \centering
  \includegraphics[width=\plotwidth]{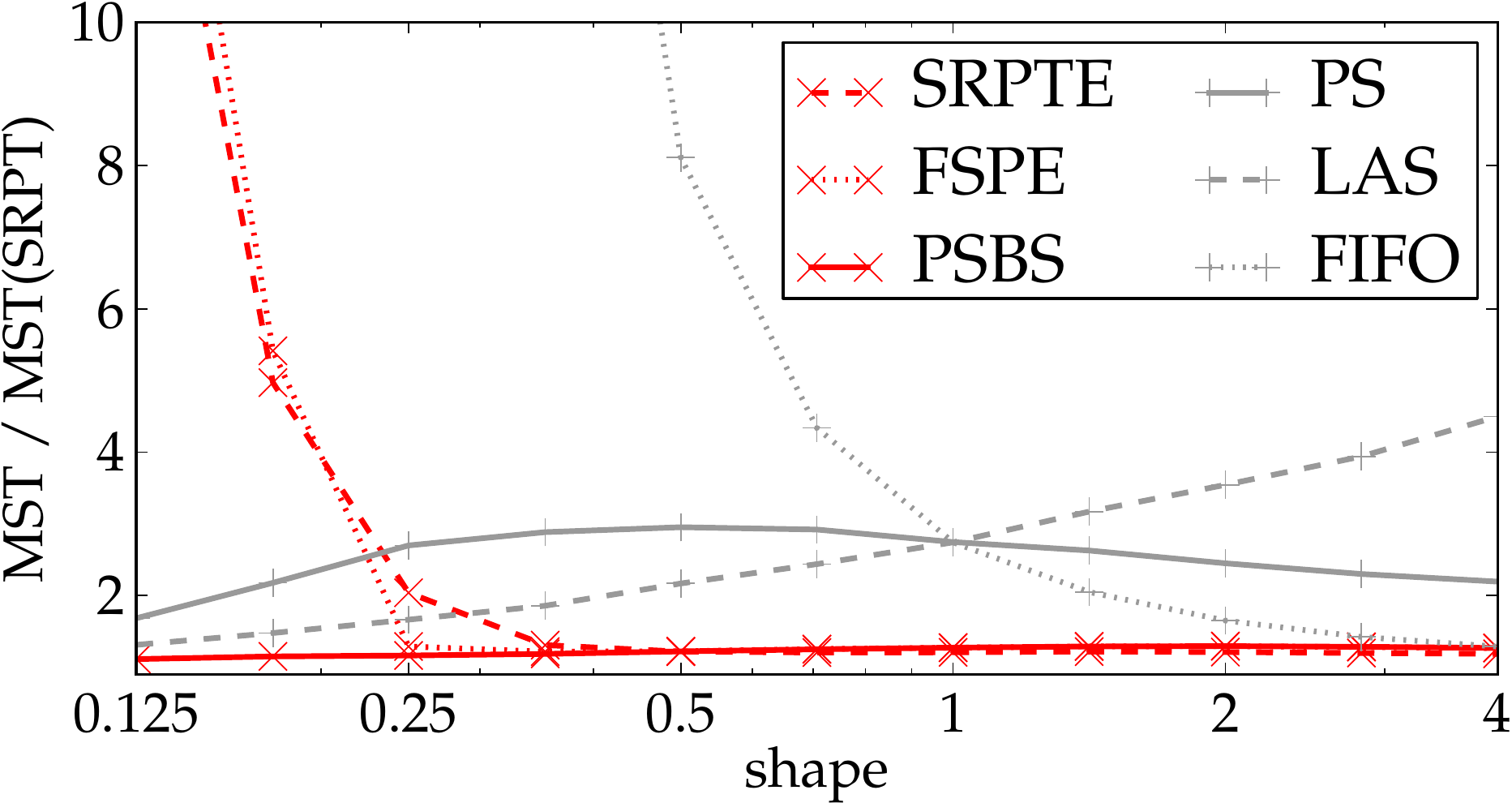}
  \caption{Impact of shape. PSBS behaves close to optimally in all cases.}
  \label{fig:shape}
\end{figure}

After validating the choice of PSBS as a generalization
of FSPE+PS, we now examine how it
performs when compared to the optimal MST that SRPT obtains. In the
following Figures, we show the ratio between the MST obtained with the
scheduling policies we implemented and the optimal one of SRPT, while
fixing sigma to its default value of 0.5.

From Fig.~\ref{fig:shape}, we see that the shape parameter is
fundamental for evaluating scheduler performance. We notice that
PSBS has \emph{almost optimal performance for all shape values
  considered},
  while SRPTE and FSPE perform
poorly for highly skewed workloads. Regarding non size-based policies,
PS is outperformed by LAS for heavy-tailed workloads (shape $< 1$) and
by FIFO for light-tailed ones having $\text{shape} > 1$; PS provides a reasonable trade-off when the
job size distribution is unknown.  When the job size distribution is
exponential (shape $= 1$), non size-based scheduling policies perform
analogously; this is a result which has been proven analytically (see
\eg the work by \citeN{harchol2009queueing} and the
references therein).    It is interesting to consider
  FIFO: in it, jobs are scheduled in series, and
  job priority is not correlated with size: indeed, the
  MST of FIFO is equivalent to the one of a random scheduler executing
  jobs in series~\cite{klugman2012loss}. FIFO can be therefore seen as
  the limit case for a size-based scheduler such as FSPE or SRPTE when
  estimations carry no information at all about job sizes; the fact
  that errors become less critical as skew diminishes can be therefore
  explained with the similar patterns observed for FIFO.

\begin{figure*}[!t]
  \centering
  \subfloat[$\text{shape}=0.25$]{
    \includegraphics[width=\threeplotwidth]{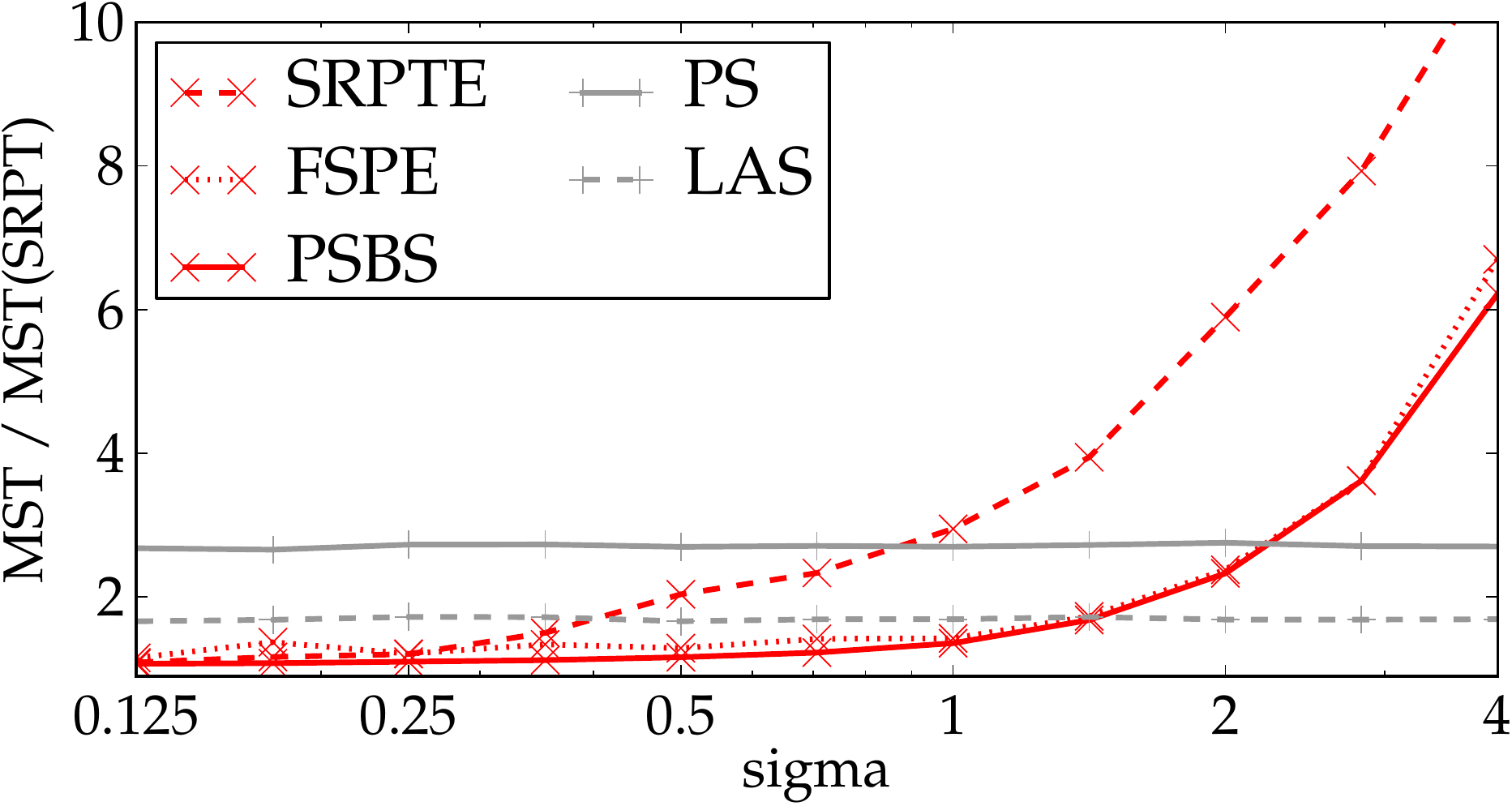}
    \label{fig:shape_0250}
  }
  \subfloat[$\text{shape}=0.177$]{
    \includegraphics[width=\threeplotwidth]{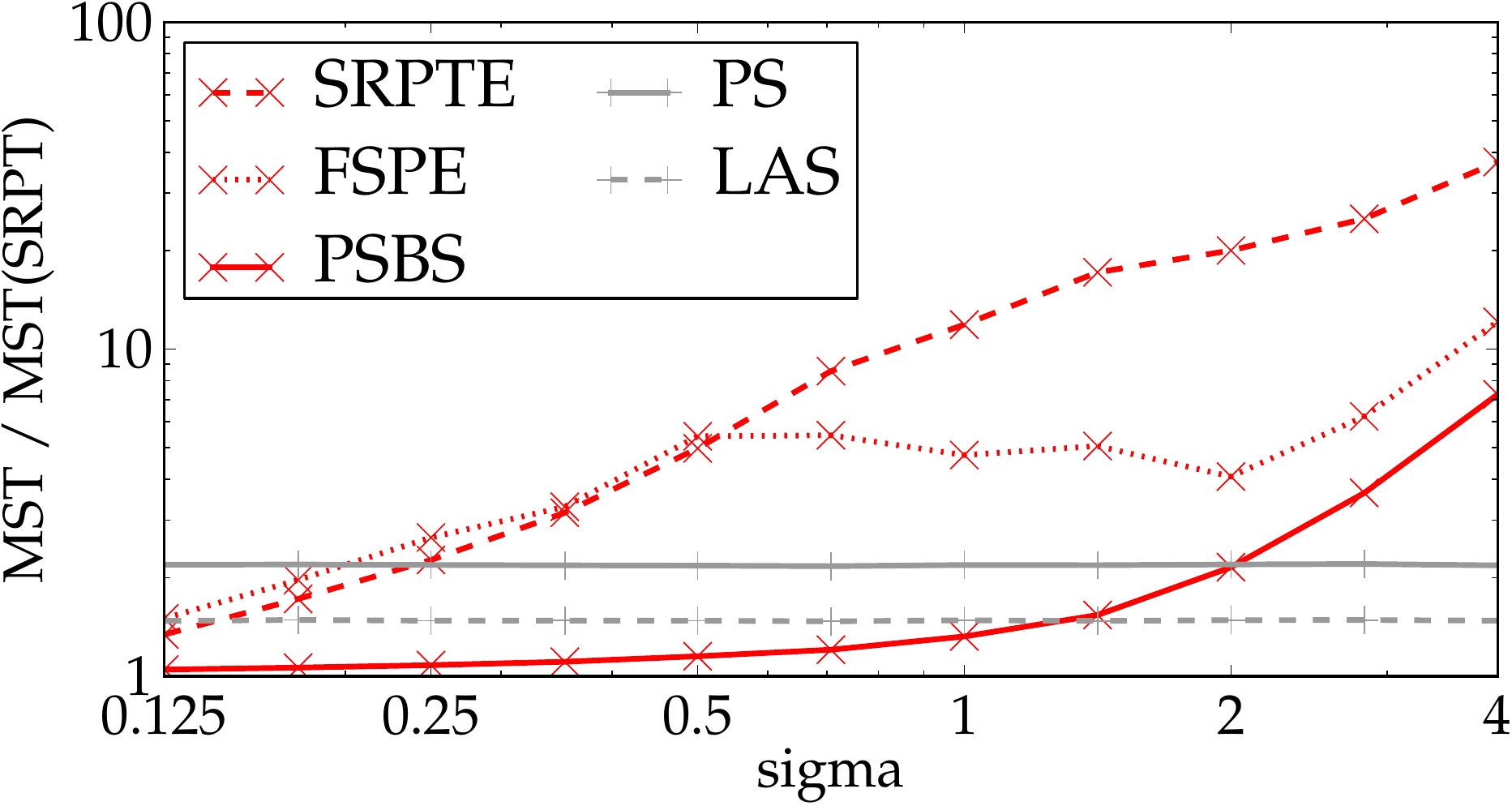}
    \label{fig:shape_0177}
  }
  \subfloat[$\text{shape}=0.125$]{
    \includegraphics[width=\threeplotwidth]{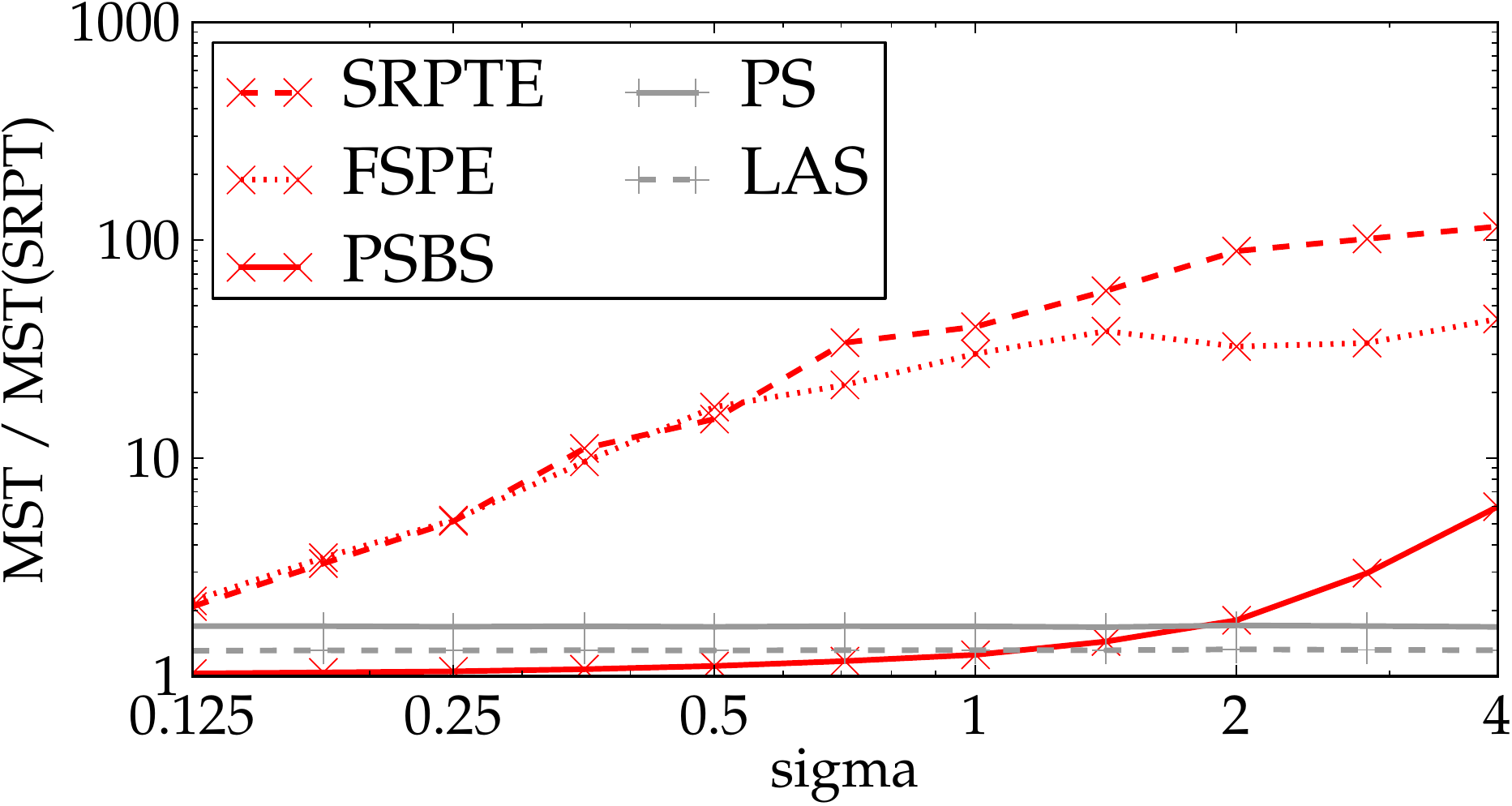}
    \label{fig:shape_0125}
  }
  \caption{Impact of error on heavy-tailed workloads, sorted by
    growing skew.}
  \label{fig:mst_heavytail}
  \label{fig:sigma}
\end{figure*}


\subsection{Impact of Sigma}
\label{sec:sigma}

The shape of the job size distribution is
fundamental in determining the behavior of scheduling algorithms, and
heavy-tailed job size distributions are those in which the
behavior of size-based scheduling differs noticeably. Because of this,
and since heavy-tailed workloads are central in the literature on
scheduling, we focus on those.

In Fig.~\ref{fig:sigma}, we show the impact of the sigma parameter
representing error for three heavily skewed workloads. In all three
plots, the values for FIFO fall outside of the plot. These plots
demonstrate that PSBS is robust with respect to errors in all the
three cases we consider, while SRPTE and FSPE suffer as the skew
between job sizes grows. In all three cases, PSBS performs better
than PS as long as sigma is lower than 2: this
corresponds to lax bounds on size estimation quality,
  requiring a correlation coefficient between job size and its
estimate of 0.15 or more.

In all three plots, PSBS performs better than SRPTE; the difference
between PSBS and FSPE, instead, is discernible only for
$\textrm{shape}<0.25$. We explain this difference by noting that, when
several jobs are in the queue, size reduction in the virtual queue of
FSPE is slow: hence, less jobs become late and therefore non
preemptable. As the distribution becomes more heavy-tailed, more jobs
become late in FSPE and differences between FSPE and PSBS become
significant, reaching differences of even around one order of
magnitude.

In particular in Fig.~\ref{fig:shape_0177}, there are areas ($0.5 <
\text{sigma} < 2$) in which increasing errors decreases (slightly)
the MST of FSPE.  This counterintuitive phenomenon is explained by the
characteristics of the error distribution: the mean of the log-normal
distribution grows as sigma grows, therefore the aggregate amount of
work for a set of several jobs is more likely to be over-estimated;
this reduces the likelihood that several jobs at once become late and
therefore non-preemptable. In other words, FSPE works better with
estimation means that tend to over-estimate job size; however, it is
always better to use PSBS, which provides a more reliable and
performant solution to the same problem.

In additional experiments -- not included due to space limitations --
we observed similar results with other error distributions; in cases
where errors tend towards underestimations, we find that the
improvements that PSBS gives over FSPE and SRPTE are even more important.

\subsection{Fairness}
\label{sec:fairness}

We now consider fairness, intending -- as discussed
in Section~\ref{sec:metrics} -- that jobs' running time should be
proportional to their size, and therefore
slowdowns should not be large.

\begin{figure}[!t]
  \centering \includegraphics[width=\plotwidth]{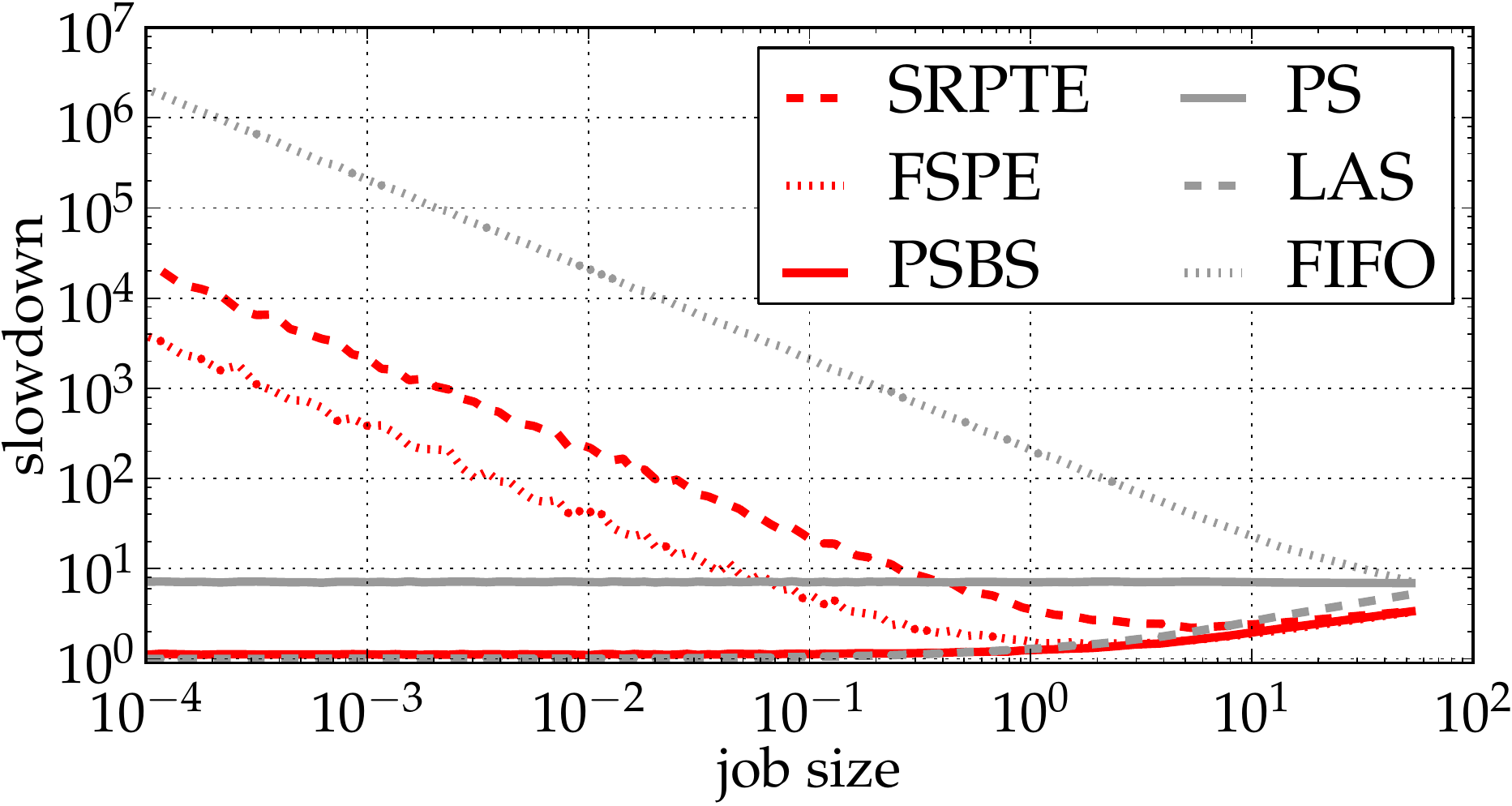}
  \caption{Mean conditional slowdown. PSBS outperforms PS, the
    scheduler often taken as a reference for fairness.}
  \label{fig:size_vs_slowdown}
\end{figure}

\paragraph*{Conditional Slowdown}

To better understand the reason for the unfairness of FIFO, SRPTE and
FSPE, in Fig.~\ref{fig:size_vs_slowdown} we evaluate \emph{mean
 conditional slowdown}, showing average slowdown (job sojourn time divided 
by job size) against job size, using our
default simulation parameters. The figure has been obtained by sorting
jobs by size and binning them into 100 job classes having similar size
and containing the same number of jobs; points plotted are obtained by averaging job size
and slowdown in each of the 100 classes.

The lines of FIFO, SRPTE and FSPE are almost parallel for smaller jobs 
because, below a certain size, \emph{job sojourn
  time is essentially independent from job size}: indeed, it depends
on the total size of older (for FIFO) or late (for SRPTE and FSPE)
jobs at submission time.

We confirm experimentally that the expected slowdown in PS is
constant, irrespectively of job size~\cite{wierman2007fairness};
PSBS and LAS, on the other hand, have close to optimal slowdown for
small jobs. PSBS has a better MST because it performs better
for larger jobs, which are more penalized in LAS.

\begin{figure}[!t]
  \centering
  \subfloat{\includegraphics[width=\plotwidth]{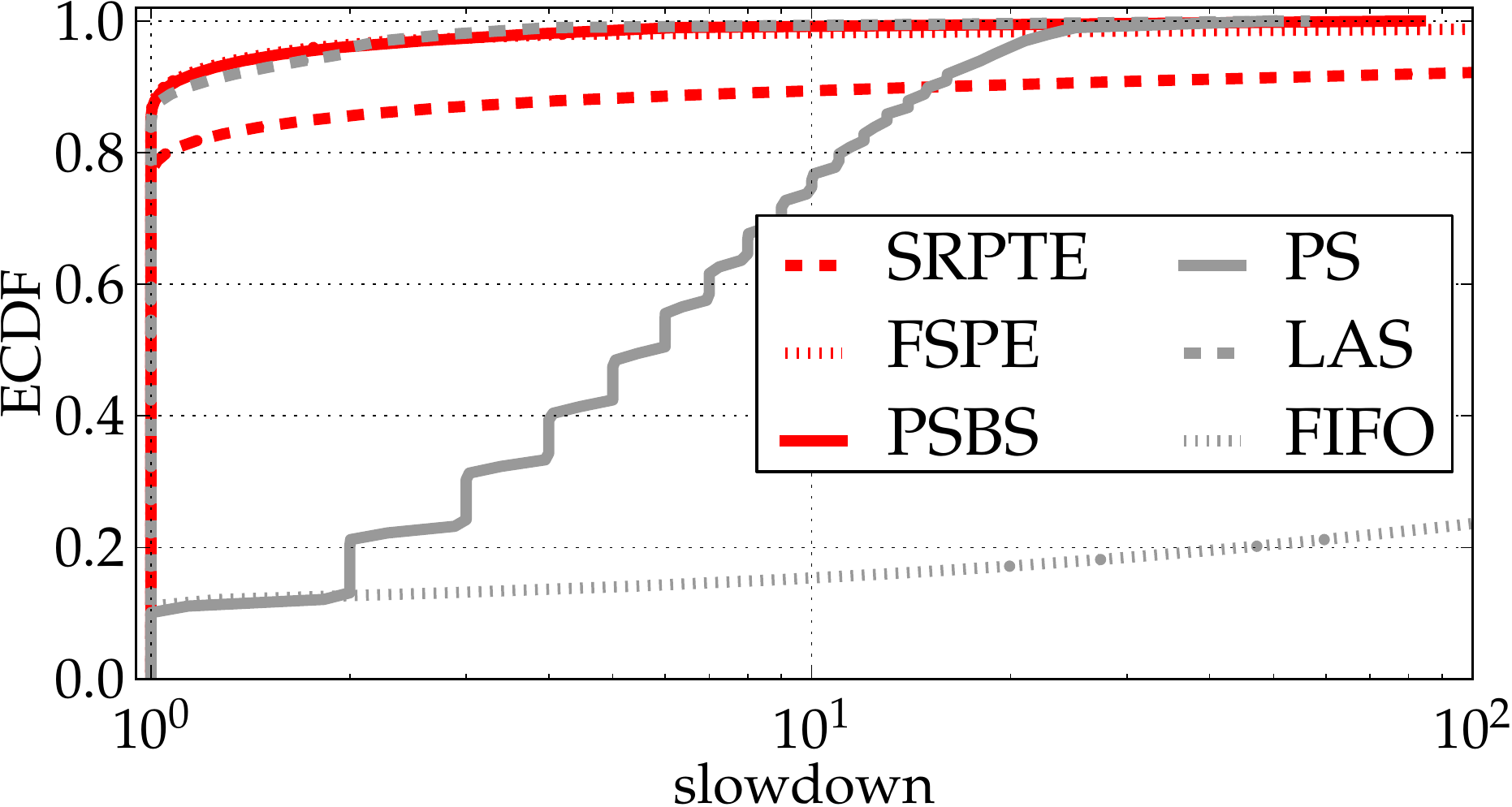}}
  \hfil
  \subfloat{\includegraphics[width=\plotwidth]{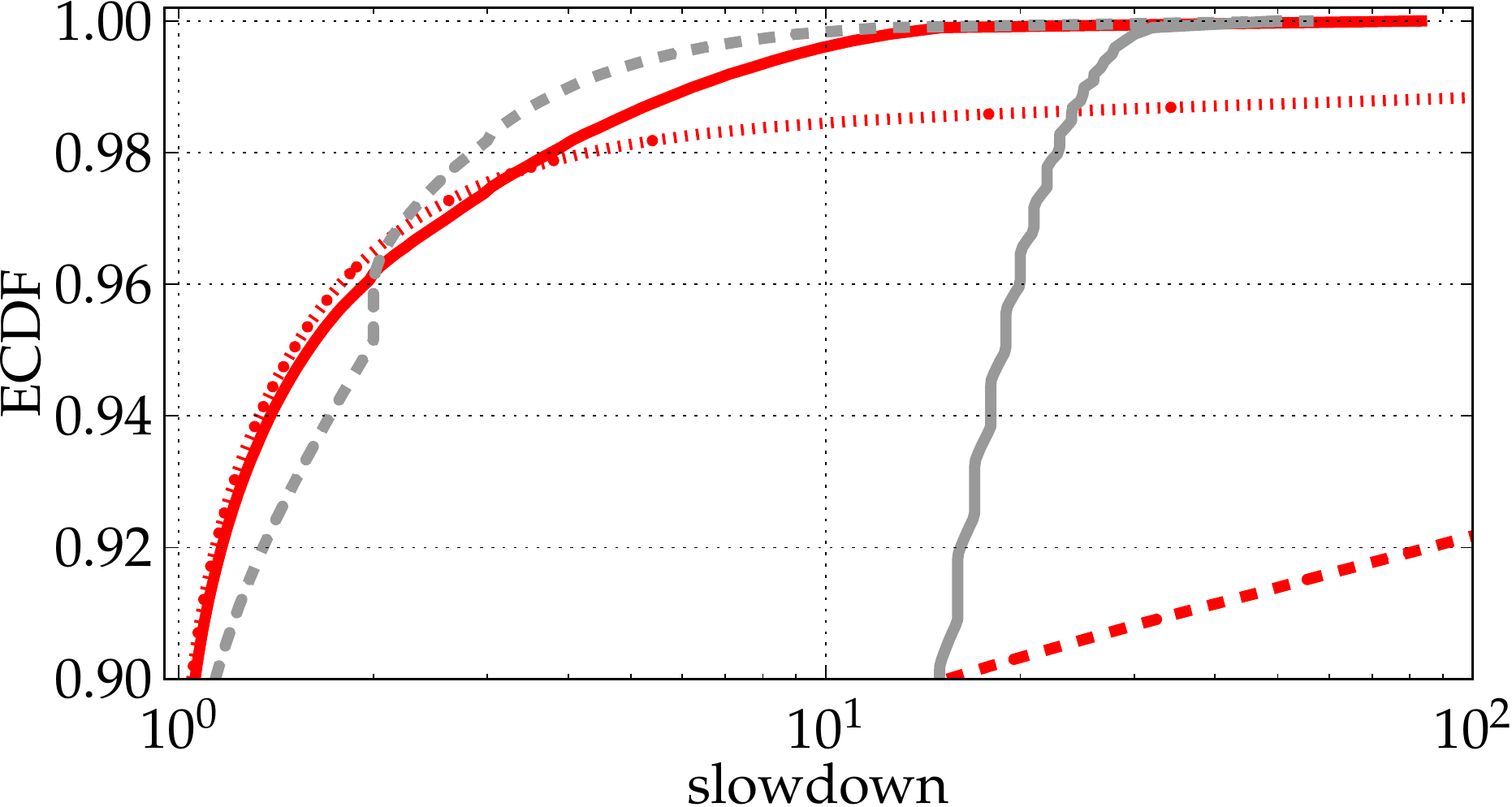}}
  \caption{Per-job slowdown: full CDF (top) and zoom on the 10\% more
    critical cases (bottom).}
  \label{fig:slowdown}
\end{figure}

\paragraph*{Per-Job Slowdown}
Our results testify that, for PSBS and similarly to
LAS, slowdown values are homogeneous across classes of job sizes:
neither small nor big jobs are penalized when using PSBS. This is a
desirable result, but the reported results are still averages:
to ensure that sojourn time is commensurate to size \emph{for
  all jobs}, we need to investigate the \emph{per-job} slowdown distribution.

In Fig.~\ref{fig:slowdown}, we plot the CDF of per-job slowdown for
our default parameters.  By serving efficiently smaller
jobs, size-based scheduling techniques and LAS manage to obtain an
optimal slowdown of 1 for most jobs. However, some jobs
experience very high slowdown: those with slowdown larger than
100 are around 1\% for FSPE and around 8\% for SRPTE.
PS, LAS, and PSBS perform well in terms of fairness, with no jobs
experiencing slowdown higher than 100 in our experiment
runs.\footnote{Fig.~\ref{fig:slowdown} plots the results of 121
  experiment runs, representing therefore 1,210,000 jobs in this
  simulation.} While PS is generally considered the reference for a
``fair'' scheduler, it obtains slightly better slowdown than LAS and
PSBS only for the most extreme cases, while being outperformed for
all other jobs.

\subsection{Job Weights}
\label{sec:priority}

\begin{figure*}[th]
   \centering
   \subfloat[$\text{shape}=0.25$]{
     \includegraphics[width=\threeplotwidth]{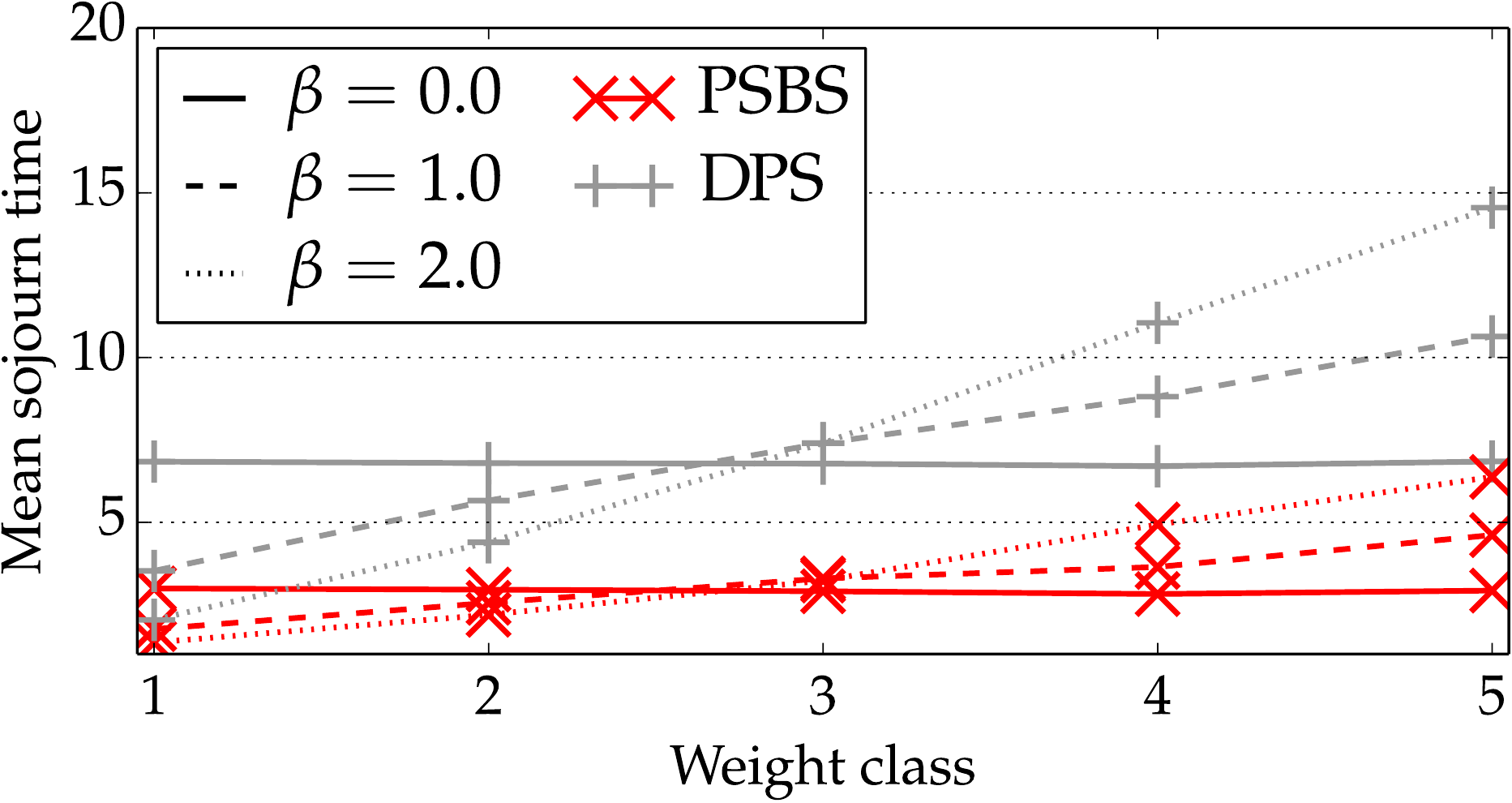}
   }
   \subfloat[$\text{shape}=1$]{
     \includegraphics[width=\threeplotwidth]{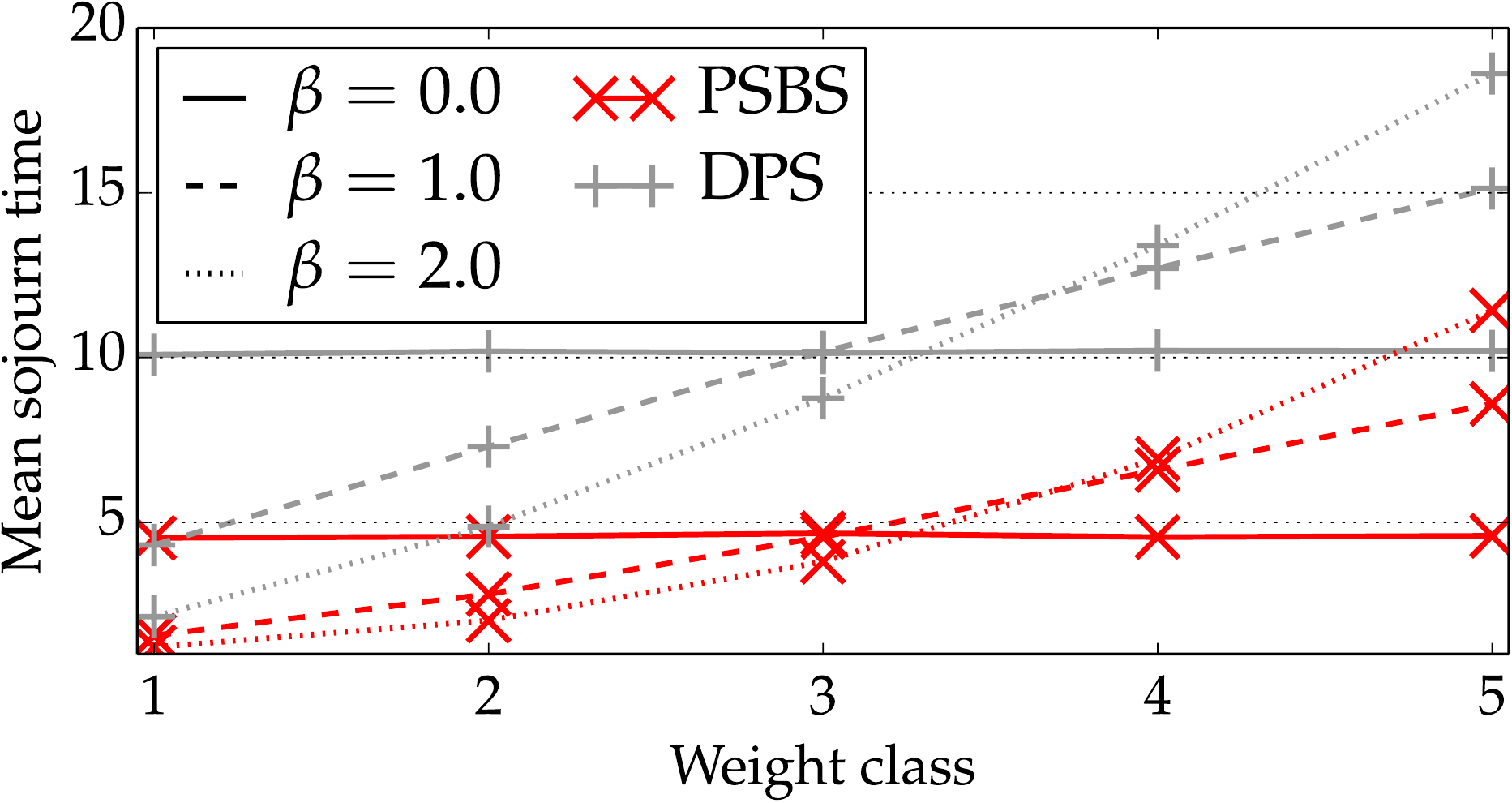}
   }
   \subfloat[$\text{shape}=4$]{
     \includegraphics[width=\threeplotwidth]{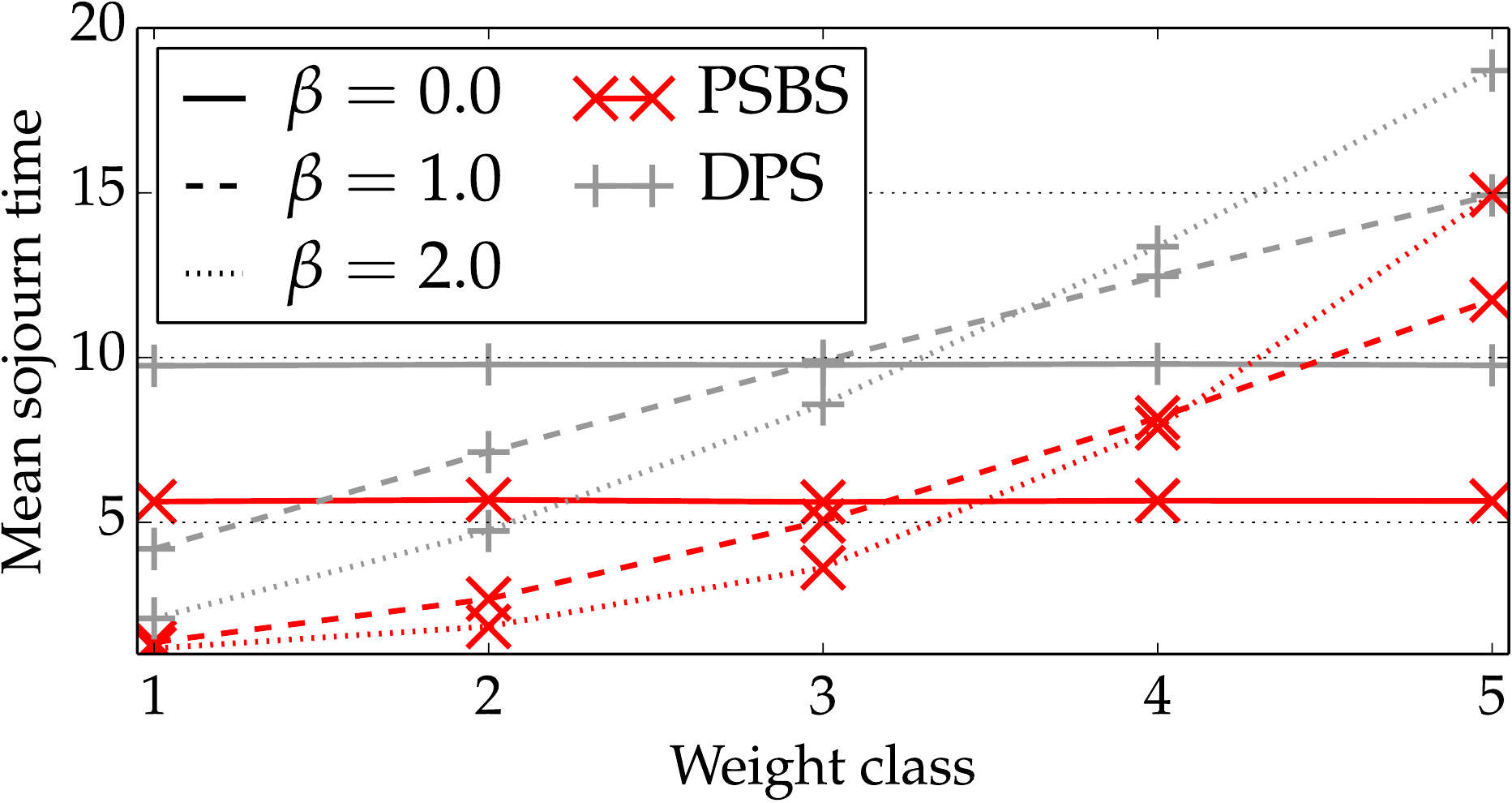}
   }
   \caption{Using weights to differentiate jobs: PSBS outperforms DPS.}
   \label{fig:priority}
\end{figure*}

We now consider how PSBS handles job weights. We consider workloads
generated with all the default values shown in
Table~\ref{table:parameters}.
Since we are not aware of representative workloads where job
priorities \emph{and} job sizes are known together, we resort to a
simple uniform distribution.
We randomly assign jobs to different weight
classes numbered from 1 to 5 with uniform probability: 
a job $i$ in weight class
$c_i$ has weight $w_i=1/c_i^\beta$, where $\beta \geq 0$ is a
parameter that allows us to tune how much we want to skew scheduling
towards favoring high-weight jobs. A $\beta = 0$ value corresponds
to uniform weights, $w_i=1$ for each job; as $\beta$ grows,
job weights differentiate so that more and more resources
are assigned to high-weight jobs.

In Fig.~\ref{fig:priority}, we plot the mean sojourn time for jobs
in each weight class. Jobs have a mean size of 1:
therefore, the best MST obtainable would be 1, which corresponds to
the bottom of the graph. We compare the results of
PSBS with those obtained by generalized processor sharing (DPS) while
using the same weights.

For workloads ranging between heavily skewed ($\text{shape}=0.25$) to
close to uniform ($\text{shape}=4$), PSBS outperforms DPS. Obviously,
$\beta=0$ leads to uniform MST between weight classes; raising the
values of $\beta$ improves the performance of high-weight jobs to
the detriment of low-weight ones. When $\beta=2$, the MST of jobs
in class 1 is already very close to the optimal value of 1; we do not
consider values of $\beta>2$ because it would impose performance
losses to low-weight jobs without significant benefits to
high-weight ones. It is interesting to point out that the trade-off
due to the choice of $\beta$ is not uniform across values of shape:
when the workload is close to uniform ($\text{shape}=4$), improvements
in sojourn times for high-weight jobs are quantitatively similar to
the losses paid by low-weight ones; this is because high-weight
jobs are likely to preempt low-weight ones with similar
sizes. Conversely, with heavily skewed workloads ($\text{shape}=0.25$)
sojourn time improvements for high-weight jobs are smaller than
losses for low-weight ones: this is because, in skewed workloads,
large high-weight jobs are likely to preempt small low-weight
ones: this results in small improvements in sojourn time for the
high-weight jobs, counterbalanced by large losses for the
low-weight ones.

\subsection{Other Settings}
\label{sec:other_settings}

\begin{figure}[!t]
  \centering
  \subfloat[$\alpha=2$]{
    \includegraphics[width=\plotwidth]{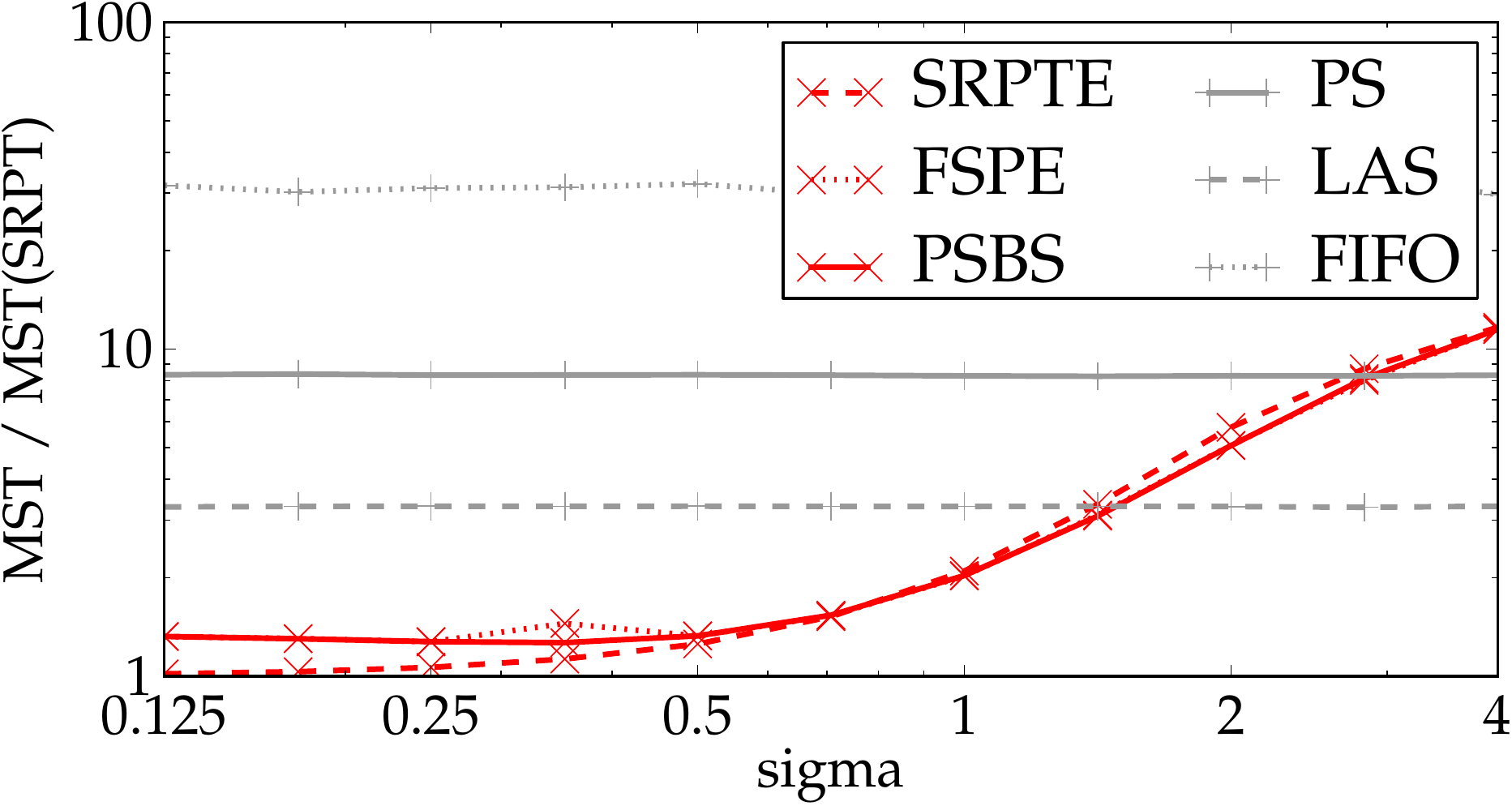}
    \label{fig:pareto2}
  }
  \hfil
  \subfloat[$\alpha=1$]{
    \includegraphics[width=\plotwidth]{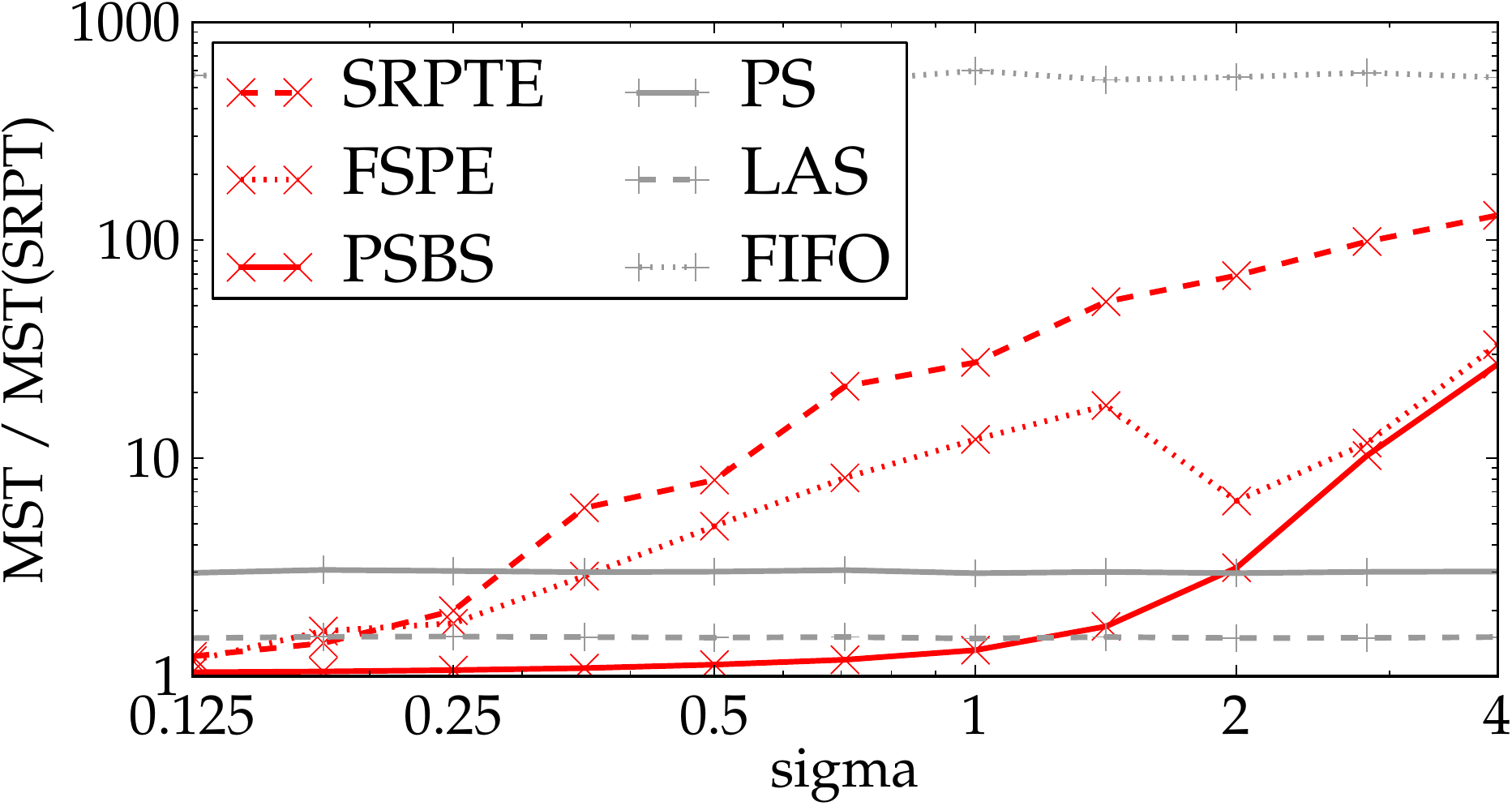}
    \label{fig:pareto1}
  }
  \caption{Pareto job size distributions, sorted by growing skew.}
  \label{fig:pareto}
\end{figure}

Until here, we focused on the sigma and shape parameters, because they
are the ones that we found out to have the most influence on scheduler
behavior. We now examine the impact of other settings
that deviate from our defaults.

\paragraph*{Pareto Job Size Distribution}
In the literature, workloads are
often generated using the Pareto distribution. To help comparing our
results to the literature, in Fig.~\ref{fig:pareto} we show results
for job sizes having a Pareto distribution, using $x_m=0$ and
$\alpha=\{1,2\}$. The results we observe for the Weibull distribution
are still qualitatively valid for the Pareto distribution; the value
of $\alpha=1$ is roughly comparable to a shape of 0.15 for the Weibull
distribution, while $\alpha=2$ is comparable to a shape of around 0.5,
where the three size-based disciplines we take into account still have
similar performance.

\paragraph*{Impact of Other Parameters}
We have studied the impact of other parameters, such as the 
load, the timeshape and the njobs, and the results are consistent with the
ones showed in the previous sections.
The interested reader can find the details in the supplemental material.

\subsection{Real Workloads}
\label{sec:real_workloads}

We now consider two real workloads to confirm that the
phenomena we observed are not an artifact of the
synthetic traces that we generated, and that they indeed apply in
realistic cases. From the traces we obtain two data points per job:
submission time and job size. In this way, we move away from the
assumptions of the $GI/GI/1$ model, and we provide results that can
account for more general cases where periodic patterns and correlation
between job size and submission times are present.

\begin{figure}[!t]
  \centering
  \includegraphics[width=\plotwidth]{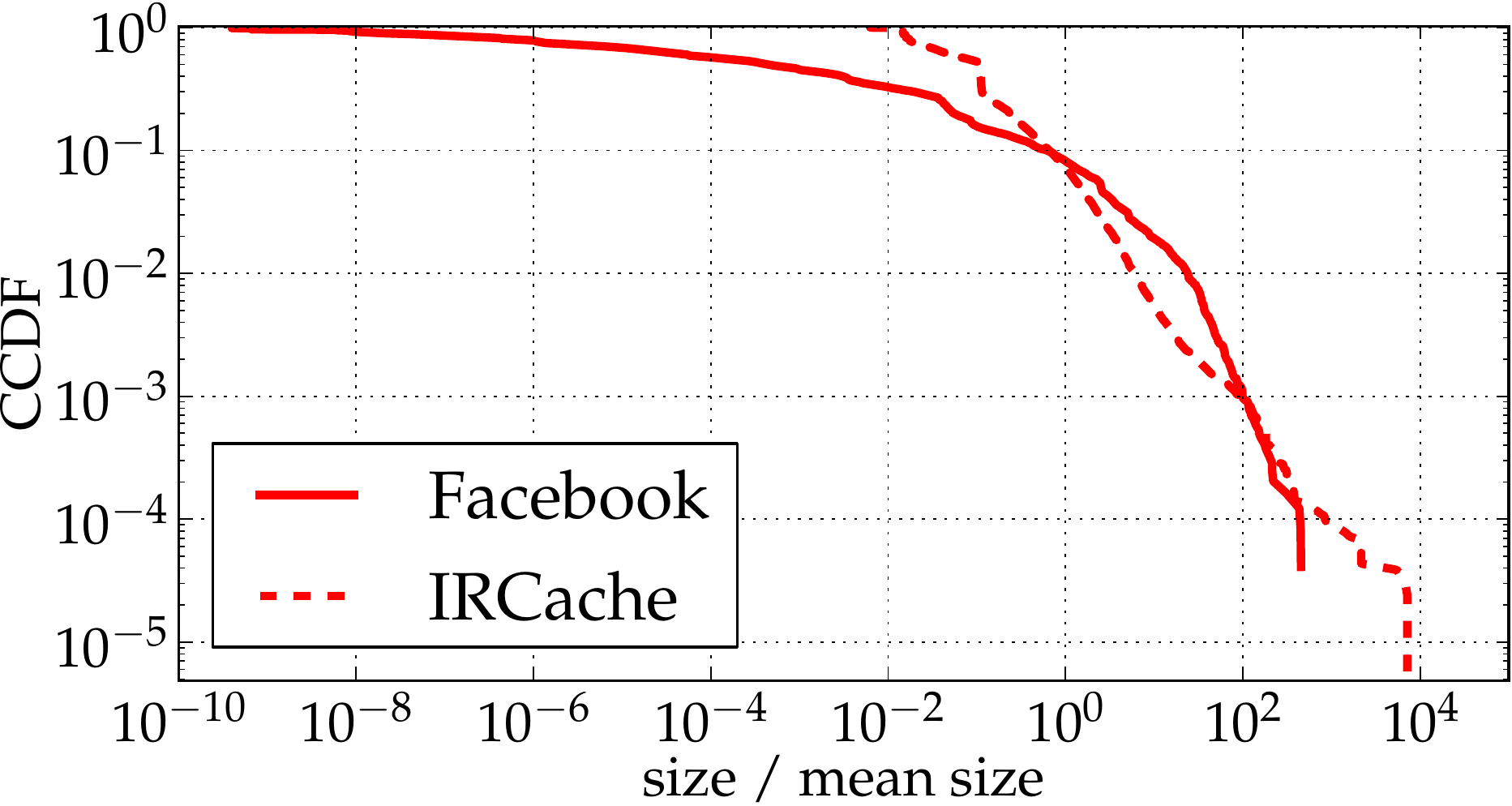}
  \caption{CCDF for the job size of real workloads.}
  \label{fig:ccdf}
\end{figure}

\paragraph*{Hadoop at Facebook}
We consider a trace from a Facebook Hadoop
cluster in 2010, covering one day of job submissions. The trace has
been collected and analized by \citeN{chen2012interactive};
it is comprised of 24,443 jobs and it is available
online.\footnote{\url{https://github.com/SWIMProjectUCB/SWIM/blob/master/workloadSuite/FB-2010_samples_24_times_1hr_0.tsv}}
For the purposes of this work, we consider the job size as the number
of bytes handled by each job (summing input, intermediate output and
final output): the mean size is 76.1 GiB, and the largest job
processes 85.2 TiB. To understand the shape of the tail for the job
size distribution, in Fig.~\ref{fig:ccdf} we plot the complementary
CDF (CCDF) of job sizes (normalized against the mean); the
distribution is heavy-tailed and the largest jobs are around 3 orders
of magnitude larger than the average size. For homogeneity with the
previous results, we set the processing speed of
the simulated system (in bytes per second) in order to obtain a load
(total size of the submitted jobs divided by total length of the
submission schedule) of 0.9.

\begin{figure}[!t]
  \centering
  \includegraphics[width=\plotwidth]{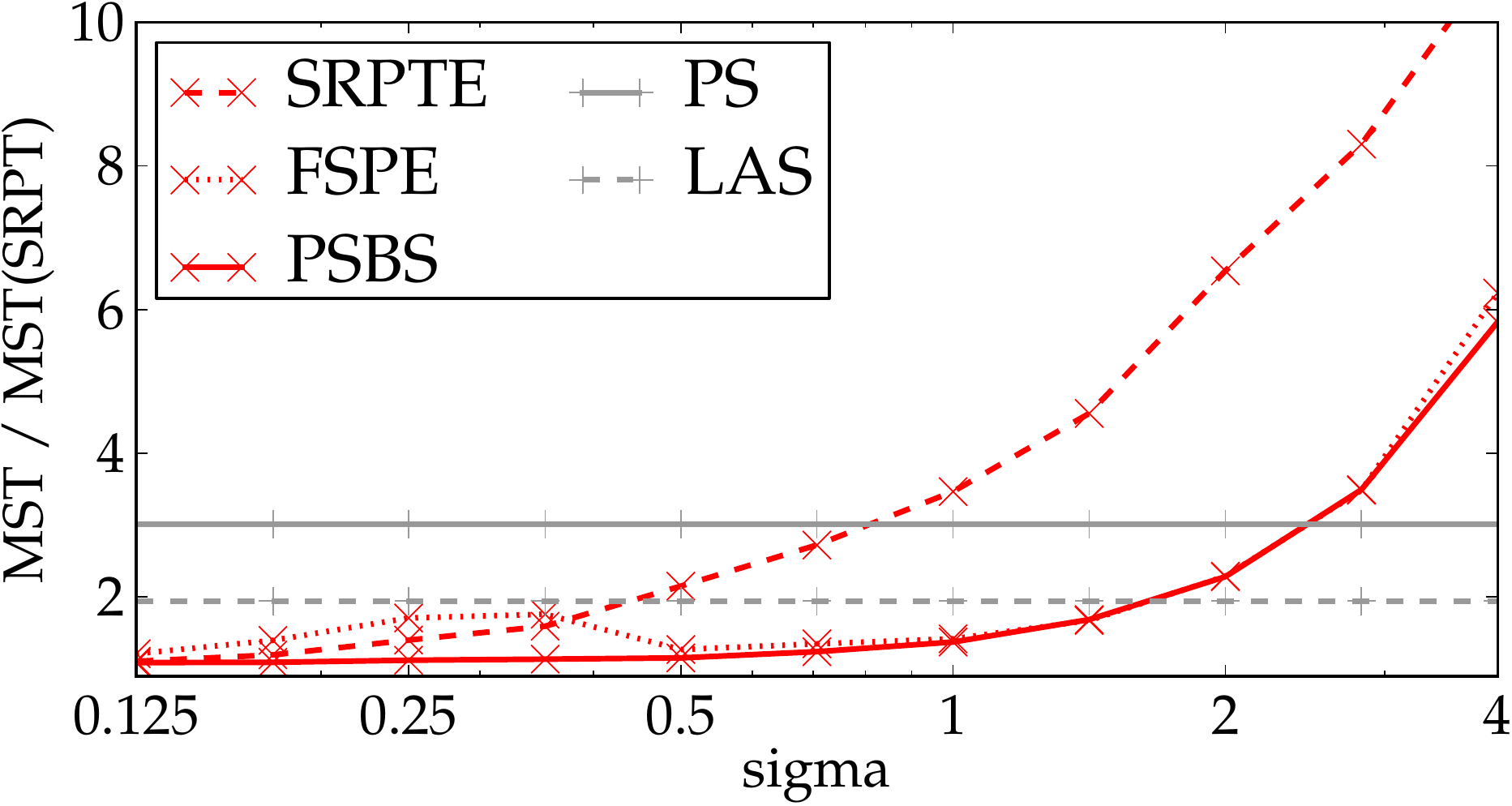}%
  \caption{MST of the Facebook workload.}
  \label{fig:fb}
\end{figure}

In Fig.~\ref{fig:fb}, we show MST, normalized against optimal MST,
while varying the error rate. These results are very
similar to those in Fig.~\ref{fig:sigma}: once again,
FSPE and PSBS perform well even when job size estimation
errors are far from negligible. 
These results show that
this case is well represented by our synthetic workloads, when
shape is around 0.25.

We performed more experiments on these traces; extensive results are
available in a technical report~\cite{dell2013simulator}.

\paragraph*{Web Cache}
IRCache\footnote{\url{http://ircache.net}} is a research
project for web caching; traces from the caches are freely
available. We performed our experiments on a one-day trace of a server
from 2007 totaling 206,914
requests;\footnote{\url{ftp://ftp.ircache.net/Traces/DITL-2007-01-09/pa.sanitized-access.20070109.gz}.}
the mean request size in the traces is 14.6KiB, while the maximum
request size is 174 MiB. In Fig. \ref{fig:ccdf} we show the CCDF of
job size; as compared to the Facebook trace analyzed previously, the
workload is more heavily tailed: the biggest requests are four orders
of magnitude larger than the mean. As before, we set the simulated
system processing speed in bytes per second to obtain a load of 0.9.

\begin{figure}[!t]
  \centering
  \includegraphics[width=\plotwidth]{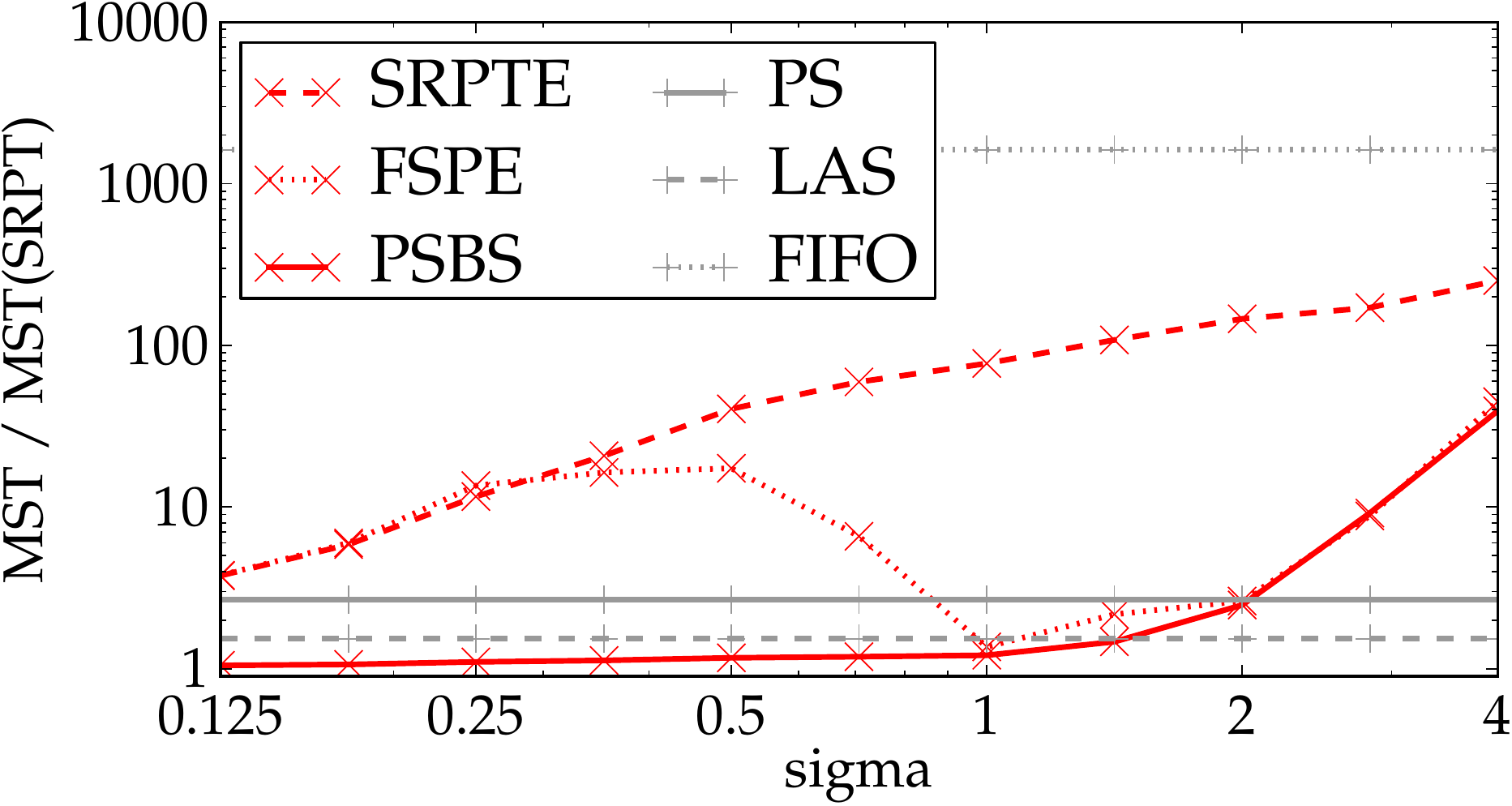}
  \caption{MST of the IRCache workload.}
  \label{fig:webserver}
\end{figure}

In Fig.~\ref{fig:webserver} we plot MST as the
sigma parameter controlling error varies. Since the job size
distribution is heavy-tailed, sojourn times are more influenced
by job size estimation errors (notice the logarithmic scale on the $y$
axis), confirming the results we have from Fig.~\ref{fig:3d}. The
performance of FSPE does not worsen monotonically as error grows, but
rather becomes better for $0.5 < \text{sigma} < 1$; this is a
phenomenon that we also observe -- albeit to a lesser extent -- for
synthetic workloads in Fig.~\ref{fig:shape_0177} and for the
Facebook workload in Fig.~\ref{fig:fb}. The explanation
provided in Section~\ref{sec:sigma} applies: since the mean of the
log-normal distribution grows as sigma grows, the aggregate amount of
work for a given set of jobs is likely to be over-estimated in total,
reducing the likelihood that several jobs at once become late and
therefore non-preemptable. Also here, we still remark that
PSBS consistently outperforms FSPE.

\section{Conclusion}
\label{sec:conclusion}

This work shows that size-based scheduling is an applicable and
performant solution in a wide variety of situations where job size is
known approximately. Limitations shown by
previous work are, in a large part, solved by the approach we took for
PSBS; analogous measures can
be taken in other preemtpive size-based disciplines.

PSBS is a generalization of FSP, and we have proven analytically that,
in the absence of errors, it dominates DPS; 
to the best of our knowledge,
PSBS is also the first $O(\log n)$ implementation of FSP.

With PSBS, system designers do not need to worry about the problems
created by job size under-estimations.
PSBS also solves a fairness problem: while FSPE and SRPTE penalize
small jobs and results in slowdown values which are not proportionate
to their size, PSBS has an optimal slowdown equal to 1 for most
small jobs.

We maintain that, thanks to its efficient implementation, solid
performance in case of estimation errors, and support for job 
weights, PSBS is a \emph{practical} size-based policy that can
guide the design of schedulers in real, complex systems. 
We argue that it is worthy to try size-based scheduling, even if 
inaccurate estimates can be produced to
estimate job sizes: our
proposal, PSBS, is reasonably easy to implement
and provides close to optimal response times and good fairness in all
but the most extreme of cases.

We released our simulator as free software; it can be reused for:
\begin{inparaenum}[(\itshape i)]
\item reproducing our experimental results; 
\item prototyping new scheduling algorithms;
\item predicting system
behavior in particular cases, by replaying traces.
\end{inparaenum}



\bibliographystyle{IEEEtranN}
{\footnotesize
\bibliography{references}
}

%

\begin{IEEEbiography}[{\includegraphics[width=1in,height=1.25in,clip,keepaspectratio]{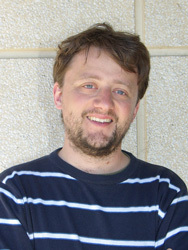}}]{Matteo Dell'Amico}
is a researcher at Symantec Research Labs; his research revolves on
the topic of distributed computing.  He received his M.S. (2004) and
Ph.D. (2008) in Computer Science from the University of Genoa (Italy);
during his Ph.D. he also worked at University College London. Between
2008 and 2014 he was a researcher at EURECOM. His research interests
include data-intensive scalable computing, peer-to-peer systems,
recommender systems, social networks, and computer security.
\end{IEEEbiography}

\begin{IEEEbiography}[{\includegraphics[width=1in,height=1.25in,clip,keepaspectratio]{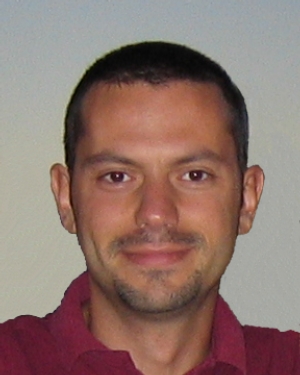}}]{Damiano Carra}
received his Laurea in Telecommunication Engineering from Politecnico di Milano, and his Ph.D. in Computer Science from University of Trento. He is currently an Assistant Professor in the Computer Science Department at University of Verona. His research interests include modeling and performance evaluation of peer-to-peer networks and distributed systems.
\end{IEEEbiography}

\begin{IEEEbiography}[{\includegraphics[width=1in,height=1.25in,clip,keepaspectratio]{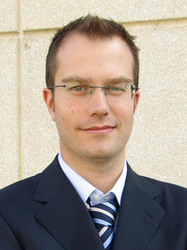}}]{Pietro Michiardi}
received his M.S. in Computer Science from EURECOM and his M.S. in Electrical Engineering from Politecnico di Torino. Pietro received his Ph.D. in Computer Science from Telecom ParisTech (former ENST, Paris). Today, Pietro is a Professor of Computer Science at EURECOM.
Pietro currently leads the Distributed System Group, which blends theory and system research focusing on large-scale distributed systems (including data processing and data storage), and scalable algorithm design to mine massive amounts of data. Additional research interests are on system, algorithmic, and performance evaluation aspects of computer networks and distributed systems. 
\end{IEEEbiography}

\newpage

\appendices
\section{Supplemental material}

\subsection{Simulator Implementation Details}

Our simulator is available under the Apache V2
license.\footnote{\url{https://github.com/bigfootproject/schedsim}}
It has been conceived with ease of prototyping in mind: for example,
our implementation of FSPE as described in Section~\ref{sec:errors}
requires 53 lines of code. Workloads can be both replayed from real
traces and generated synthetically.

The simulator has been written with a focus on computational
efficiency. It is implemented using an event-based paradigm, and we
used efficient data structures based on B-trees.\footnote{\url{http://stutzbachenterprises.com/blist/}}
As a result of these choices, a ``default'' workload of 10,000 jobs is
simulated in around half a second, while using a single core in our
2011 laptop with an Intel T7700 CPU. We use IEEE 754 double-precision
floating point values to represent time and job sizes.

\subsection{Additional experiments}

\begin{figure}[!b]
  \centering
  \subfloat[Varying load.]{
    \includegraphics[width=\plotwidth]{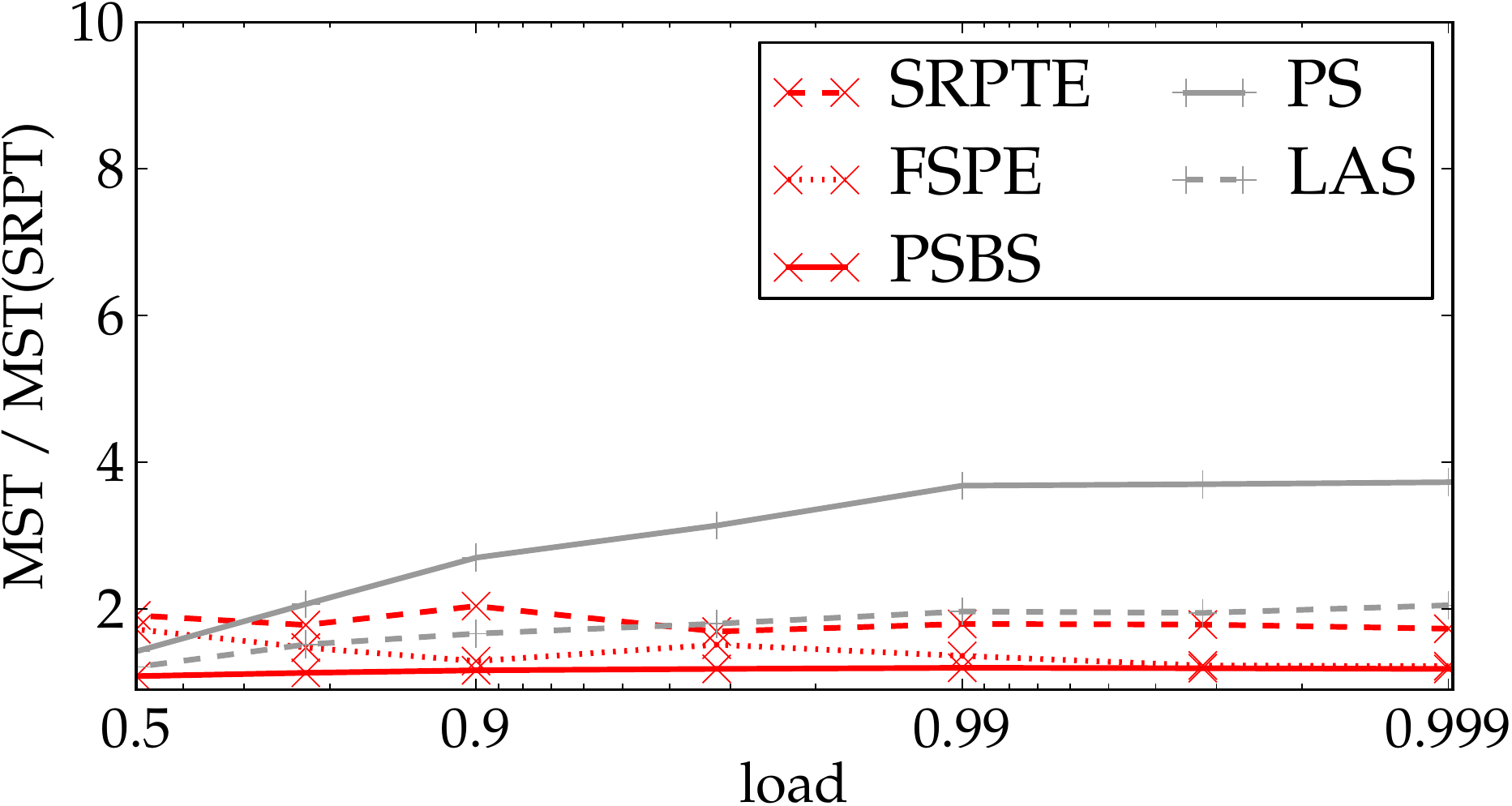}
    \label{fig:load}
  }
  \hfil
  \subfloat[Varying timeshape.]{
    \includegraphics[width=\plotwidth]{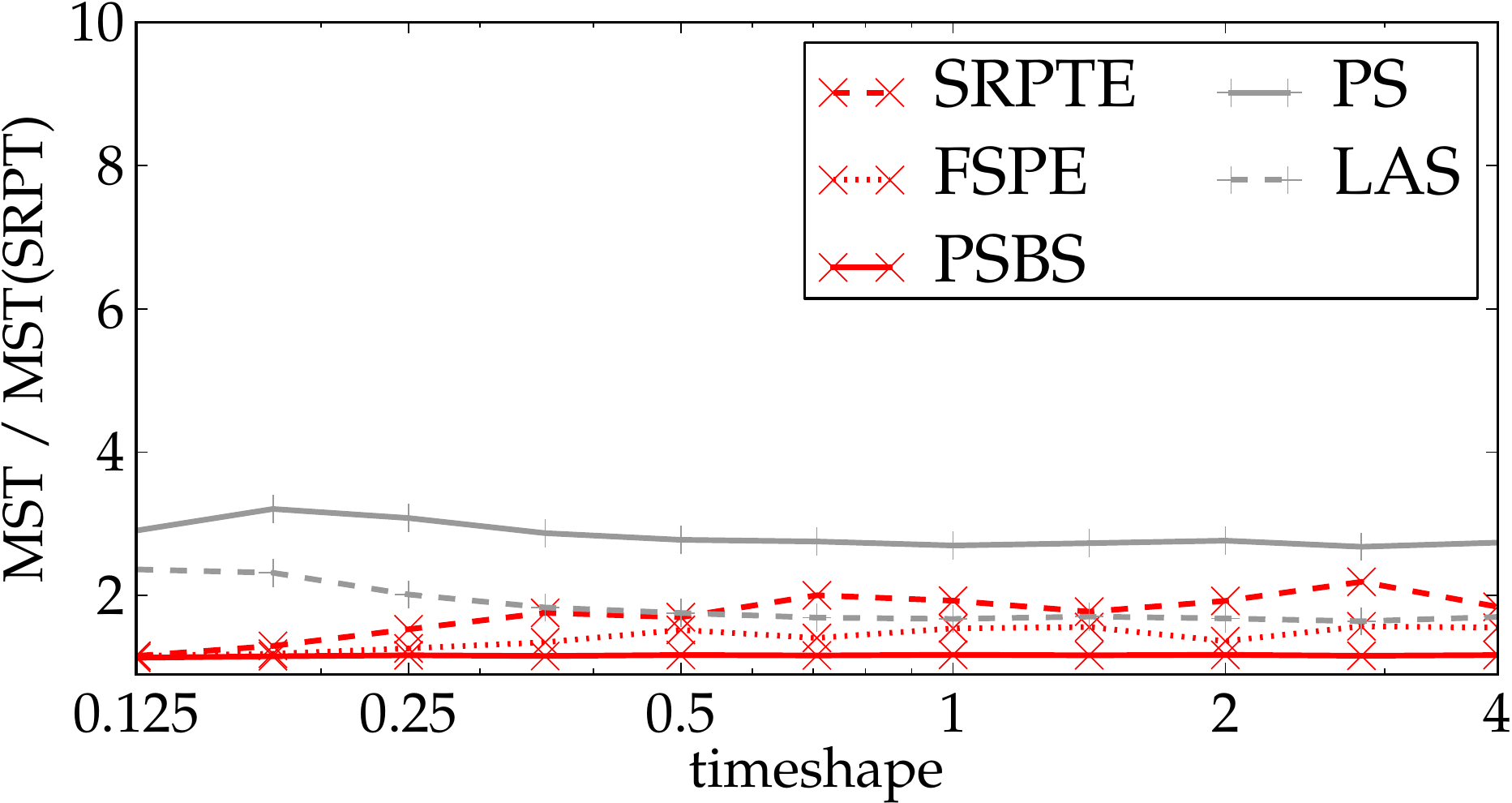}
    \label{fig:timeshape}
  }
  \caption{Impact of load and timeshape.}
  \label{fig:sim}
\end{figure}

In Fig.~\ref{fig:sim}, we show the impact of load and
timeshape, keeping sigma and shape at their default
values. 
Fig.~\ref{fig:load} shows that performance of size-based scheduling
protocols is not heavily impacted by load, as the ratio between the
MST obtained and the optimal one remains roughly constant (note that
the graph shows a ratio, and not the absolute values which increase as the 
load increases); conversely,
size-oblivious schedulers such as PS and LAS deviate more from optimal
as the load grows.
Fig.~\ref{fig:timeshape} shows the impact of changing the timeshape
parameter: with low values of timeshape, job submissions are bursty
and separated by long pauses; with high values job submissions are
more evenly spaced. We note that size-based scheduling policies respond
very well to bursty submissions where several jobs are submitted at
once: in this case, adopting a size-based policy that focuses all the
system resources on the smallest jobs pays best; as the intervals
between jobs become more regular, SRPTE and FSPE become slightly less
performant; PSBS remains close to optimal.

\begin{figure}[htb]
  \centering
  \subfloat[Varying load.]{
    \includegraphics[width=\threeplotwidth]{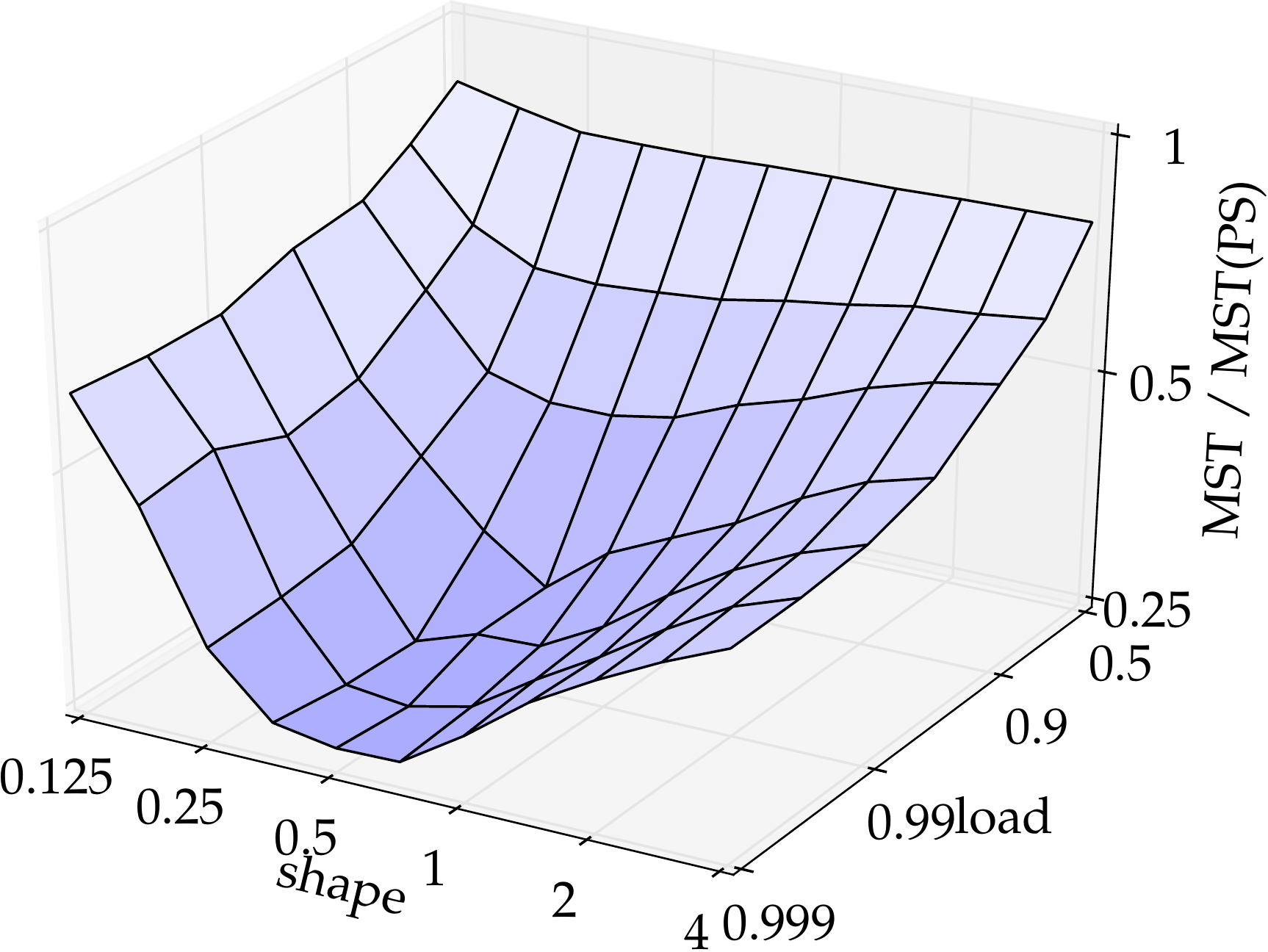}
  }
  \hfil
  \subfloat[Varying timeshape.]{
    \includegraphics[width=\threeplotwidth]{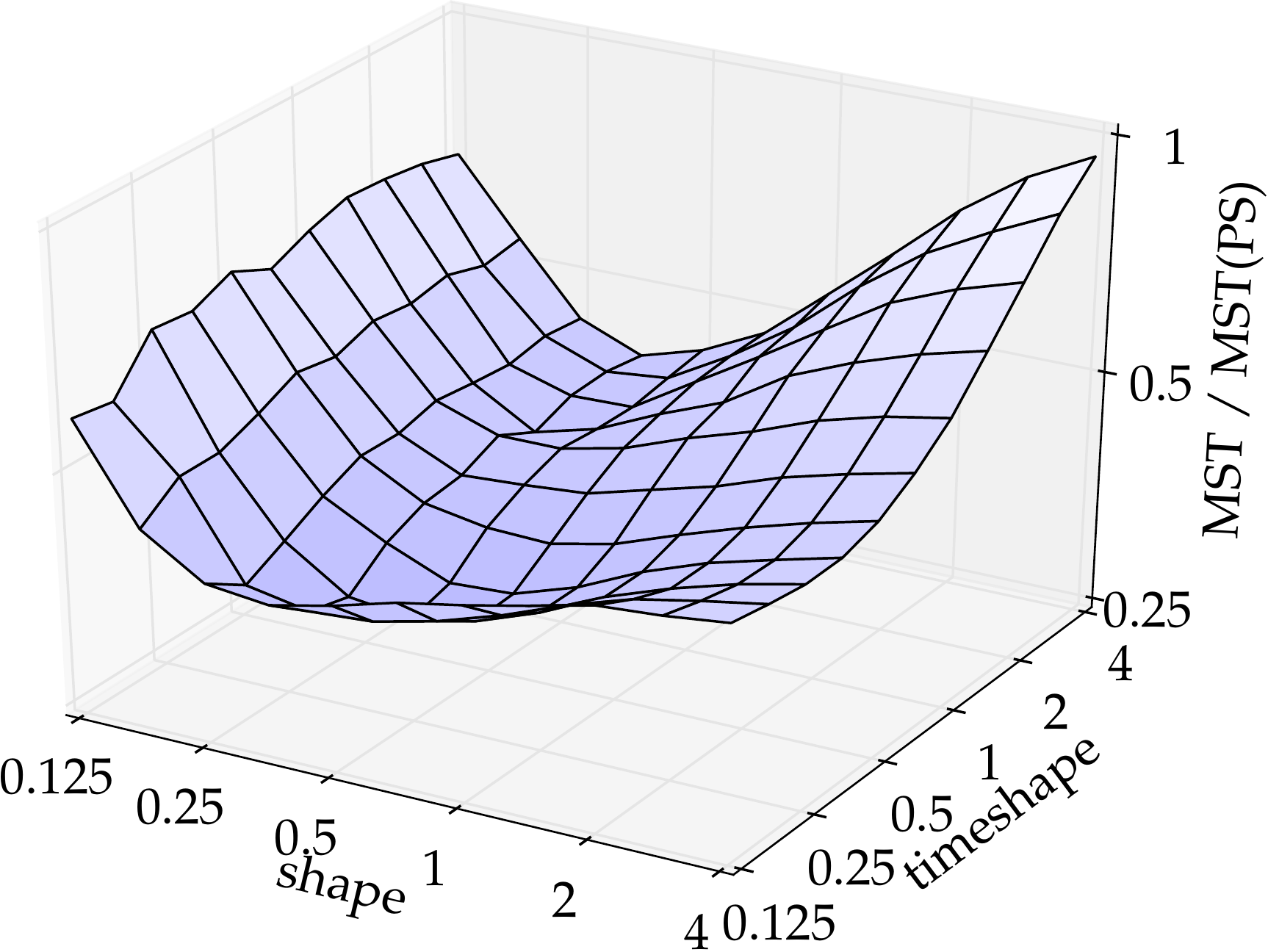}
  }
  \hfil
  \subfloat[Varying njobs.]{
    \includegraphics[width=\threeplotwidth]{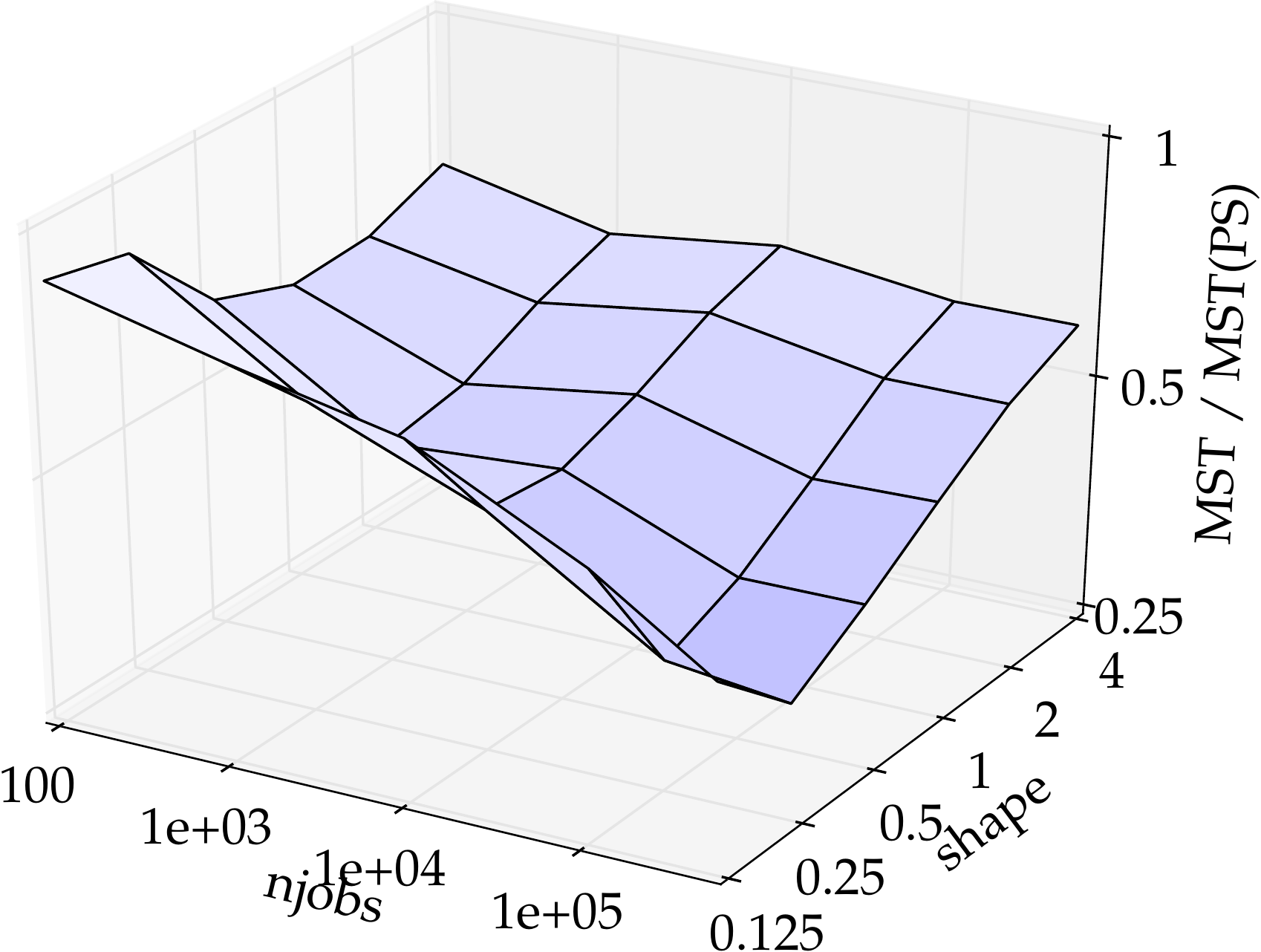}
  }
  \caption{PSBS against PS.}
  \label{fig:3d_fspe+ps_parameters}
\end{figure}


With Fig.~\ref{fig:3d_fspe+ps_parameters}, we show that the results
of Fig.~\ref{fig:sim} generalize to other parameter choices: by
letting shape vary together with load, timeshape and njobs, we notice
that PSBS always performs better than PS. The V-shaped pattern, where the
difference in performance between the two schedulers is larger for
``central'' values of the shape parameter is essentially caused by
PS performing closer to optimal for extreme values of the shape parameter, 
as we can see in Fig.~\ref{fig:shape}.

\end{document}